\documentclass[11pt, epsfig]{article}
\usepackage{epsfig, amsmath, amssymb, amsthm, times}
\usepackage{graphicx}

\usepackage{color} 

\newcommand{\R}{{\mathbb R}}

\newcommand{\Z}{{\mathbb Z}}

\usepackage[mathcal]{eucal}

\newtheorem {thm}{Theorem}[section]

\theoremstyle{defintion}

\theoremstyle{remark}
\newtheorem{rem}[thm]{Remark}

\theoremstyle{example}

\theoremstyle{assumption}

\def\T{\operatorname{\mathbb T}}

\def\lbl{\label}
\def\be{\begin{equation}}
	\def\ee{\end{equation}}
\def\p{\partial}
\newcommand{\1}{\mathrm{i}\mkern1mu}

\title{A global bifurcation organizing rhythmic activity in a coupled network}
\author{Georgi S. Medvedev\thanks{Department of Mathematics, Drexel
		University, Philadelphia, PA 19104, 
		{\tt medvedev@drexel.edu}} \and Matthew S. Mizuhara\thanks{Department of 
		Mathematics and Statistics,
		The College of New Jersey, Ewing, NJ 08628,
		{\tt  mizuharm@tcnj.edu}}
	\and Andrew Phillips\thanks{Department of Physics, Drexel University,
		Philadelphia, PA 19104, 
		{\tt apr384@drexel.edu}} }
\begin{document}
	\maketitle
	
	\begin{abstract}
		We study a system of  coupled phase oscillators near a saddle-node 
		on an invariant circle bifurcation  and driven by random intrinsic
		frequencies. Under the variation of control parameters,
		the system undergoes a phase transition changing the qualitative properties of 
		collective dynamics. Using the Ott-Antonsen reduction and geometric techniques for
		ordinary differential equations, we identify a heteroclinic bifurcation in a family of vector
		fields on a cylinder, which explains the change in collective dynamics.
		Specifically, we show that the heteroclinic bifurcation separates
		two topologically distinct families of limit cycles: contractible limit cycles before the
		bifurcation from  noncontractibile ones after the bifurcation. Both families are stable
		for the model at hand.
	\end{abstract}

	{ The Kuramoto model (KM) of coupled phase oscillators provides an
		important paradigm for studying collective dynamics in systems ranging
		from neuronal networks to swarms of fireflies to power grids. The classical
		KM features a remarkable phase transition separating stable
		mixing dynamics from gradual build-up of synchronization. During the latter
		phase, the oscillators form a cluster whose coherence (measured by the order
		parameter) remains approximately constant and increases with the coupling
		strength. If the uniformly rotating phase oscillators in the KM
		are replaced by those 
		close a saddle-node on invariant circle bifurcation,
		the order parameter does not stay constant anymore. Instead, it undergoes slow-fast
		oscillations. Furthermore,
		for larger values of the coupling strength the system undergoes  a
		phase transition, which changes the character of oscillations
		qualitatively. Previous studies based in part on numerical bifurcation techniques
		revealed a rich bifurcation structure of the modified KM. In this paper,
		we use the Ott-Antonsen reduction and qualitative methods for ordinary differential
		equations to study collective dynamics in the modified model with an emphasis
		on the slow-fast oscillations of the order parameter. We analytically locate the
		Andronov-Hopf and Bogdanov-Takens bifurcations for the system of equations
		governing the order parameter and identify the relevant normal form for the
		Bogdanov-Takens bifurcations. Furthermore, we relate the phase transition in the
		modified model to a nonlocal bifurcation in one parameter families of vector
		fields on a cylinder. The results of this work show that the slow-fast oscillations
		of the order parameter in the modified KM are shaped by the
		proximity to both  Bogdanov-Takens and heteroclinic bifurcations.
	}
	
	\section{Introduction}\label{sec.intro}
	\setcounter{equation}{0}
	
	The KM plays a special role in the theory synchronization. It provides a
	framework for studying synchronization and other forms of collective dynamics in systems
	of coupled oscillators with random parameters. Despite its analytical simplicity,
	studies of the KM revealed and helped to understand some very nontrivial
	phenomena in collective dynamics, which are relevant to a range of models in physics
	and biology (e.g., the onset of synchronization and chimera states \cite{Kur75, KurBat02, Str00}).
	Motivated by models featuring type I excitability in mathematical biology,
	we modify the KM by placing
	the individual oscillators close a saddle node on an invariant circle bifurcation.
	Specifically, we consider the following coupled system
	\begin{equation}\label{KM}
		\dot\theta_i=1+\omega_i -\cos\theta_i
		+Kn^{-1}\sum_{j=1}^n\sin\left(\theta_j-\theta_i\right),\quad i\in [n]:=\{1,2,\dots,n\},
	\end{equation}
	where $\theta_i:\R^+\to \T:=\R/2\pi\R$ is the state of oscillator $i$ at time $t$ and
	$K\ge 0$ is the coupling strength. Parameter $\omega_i$ controls the frequency
	of oscillator $i$, if $\omega_i>0$, or defines the excitation threshold otherwise.
	We assume that $\omega_i$'s are sampled from a unimodal probability distribution
	with density $g(\omega)$.
	Equation \eqref{KM} fits into the framework of the coupled
	active rotators model considered by Shinomoto and Kuramoto in \cite{ShiKur86}.
	In contrast to our setting, they used identical oscillators
	(i.e., $\omega_i=\operatorname{const}$) albeit forced by small noise. Different variants
	of the active rotators model with and without noise were studied more recently by other authors
	\cite{AceBon98, TesSciCol2007, ZaksNeiman2003, KlinFran2015, KlinFran2019}.
	We comment on the relation of our
	findings to the results in these papers below.

	To take a full advantage
	of the Ott-Antonsen Ansatz \cite{OttAnt08} below, we restrict to the Lorentzian distribution
	with density
	\begin{equation}\label{Lorentz}
		g(\omega)=\frac{\delta}{\pi \left(\omega-\epsilon^2\right)^2+\delta^2},
	\end{equation}
	where $\epsilon^2$ and $\delta>0$ are the location and scale parameters respectively.
	Lorentzian distribution is commonly used in the studies of the collective dynamics in the Kuramoto
	model and related systems, because it fits nicely into the Ott-Antonsen Ansatz
	(see, e.g., \cite{OttAnt09, OmeLai2022}). There has been a concern that
	due to its special properties (lack of finite moments) models based on Lorentzian distribution
	may feature nongeneric scenarios \cite{LafColTor2010}. We checked that 
	in qualitative form the bifurcation scenarios reported in this paper hold for Gaussian distribution,
	and, thus, are
	relevant to a large class of models based on unimodal probability distributions.
	If typical realizations of $\omega_i$ are $O(1)$ and positive then the dynamics of \eqref{KM} is
	practically the same as in the classical KM. Specifically, for small
	$K$ the oscillators are
	highly incoherent. Starting with a
	certain critical value of $K$, the coherence gradually increases \cite{Str00}.
	Thus, in this paper, we focus on the regime when
	typical values of $\omega_i's$ are small, i.e., when the individual oscillators are close a
	saddle-node bifurcation.

	To describe the collective dynamics of \eqref{KM} we need to
	remind the reader the definition of the Kuramoto's order parameter:
	\begin{equation}\label{order-p}
		h_n=n^{-1} \sum_{j=1}^n e^{\1 \theta_j},\quad
		h_n=\rho_n e^{\1\phi_n}.
	\end{equation}
	The modulus, $\rho_n=|h_n|$, and the  argument, $\phi_n=\operatorname{arg}h_n$, of the complex order parameter
	yield the degree of coherence and the position of the 
	center of mass of the population of oscillators on the unit circle respectively.
	The combination of these two quantitates as functions of time provides a good description
	of the collective dynamics. In the classical Kuramoto, after some transients the modulus
	of the order parameter approaches a steady value (up to $O(n^{-1/2})$ fluctuations) while
	the argument drifts with approximately constant velocity. In contrast, in the simulations
	of the modified KM \eqref{KM} we observe very nonuniform slow-fast oscillations
	in the modulus and the argument of the complex order parameter (Figure~\ref{f.order}).
	
	We now turn to describe the salient features of the  collective dynamics of \eqref{KM}.
	There are three main regimes in the
	system dynamics:
	\begin{description}
		\item[I]
		For small values of $K$, the individual oscillators
		rotate in the counterclockwise direction but
		the overall distribution of oscillators on a unit circle remains
		practically stationary. The distribution has a peak at $\theta=0$
		reflecting the fact that the oscillators slow down in a neighborhood
		of $\theta=0$ and thus spend more time there (see  Figure~\ref{f.snapshots}\textbf{a}).
		The order parameter lies on the real axis very close to $1$ and stays
		approximately constant (see inset in  Figure~\ref{f.snapshots}\textbf{a}).
		\item[II]
		There are two critical values of $K$: $0<K_{AH}<K_{HC}$, which will be shown  below 
		to correspond to an Andronov-Hopf bifurcation and a heteroclinic bifurcation respectively.
		For $K\in (K_{AH}, K_{HC})$ the order parameter undergoes slow-fast oscillations
		(see Figure~\ref{f.order}). Importantly, the argument of the order parameter stays
		between $\frac{-\pi}{2}$ and $\frac{\pi}{2}$ (see the insets in Figures~\ref{f.snapshots}\textbf{b},\textbf{c}).
		This means that the center of mass of the population of oscillators oscillates around $\theta=0$.
		The modulus of the order parameter varies between values close to $0$ and $1$ 
		spending more time in the region near $1$ (Figure~\ref{f.order}).
		Two snapshots of the distribution of the oscillators on a unit circle are shown in
		Figure~\ref{f.snapshots}\textbf{b} (coherent phase) and Figure~\ref{f.snapshots}\textbf{c} (incoherent phase).
		\item[III]
		Two representative episodes of the system dynamics for larger values of $K$ are shown in
		Figure~\ref{f.snapshots}\textbf{d} and \textbf{e}. As before, the oscillators slow down and
		accumulate as they approach the origin and accelerate and spread around once they have
		passed it. However, there is an important distinction from the previous regime.
		It is quite pronounced in numerical simulations and can also be seen from the static snapshots
		in Figure~\ref{f.snapshots}. Recall that for small
		values of $K$ ($K<K_{HC}$) the order parameter never leaves the right half plane.
		For larger $K$ ($K>K_{HC}$),  on the other hand, the order parameter makes a full revolution around the origin
		in one cycle of oscillations (see the insets in Figure~\ref{f.snapshots}~\textbf{d},\textbf{e}).
		The change in the qualitative character of oscillations can be seen from the timeseries of
		$\phi_n$ for two different values of $K$ in Figure~\ref{f.order}~\textbf{b}. Note that for small $K$,
		$\phi_n$ undergoes small oscillations, while for large values of $K$, the range of $\phi_n$ covers
		the entire circle. One can compare this to small oscillations versus full swing revolutions
		of the pendulum. Below, we show that indeed the two regimes are qualitatively (topologically)
		distinct and are separated by a bifurcation.
	\end{description}

	In contrast to the classical scenario of transition to synchronization in the original KM, where
	one observes the formation of a single coherent cluster drifting uniformly around the unit
	circle \cite{Kur75, Str00}, in the modified model \eqref{KM}, we see pronounced oscillations in the argument of the
	order parameter (Figure~\ref{f.snapshots} \textbf{b}-\textbf{e}). The amplitude of these oscillations increases until the oscillations are
	transformed into a rotational motion of the center of mass of the
	population of the oscillators around the unit circle (Figure~\ref{f.snapshots} \textbf{d},\textbf{e}).
	Another salient feature of this transtion is the slow-fast character of the oscillations of the
	order parameter (Figure~\ref{f.order}).

	Periodic regimes in macroscopic dynamics in the model of active rotators were
	described already in \cite{ShiKur86}. These regimes were found by numerical simulations.
	Pulsating oscillations of the order parameter similar to those shown in Figure~\ref{f.order}
	were 
	reported in the subsequent studies of  the coupled active rotators 
	with and without
	noise \cite{KlinFran2015, TesSciCol2007, ZaksNeiman2003}. In \cite{KlinFran2015, TesSciCol2007},
	the authors derived a system of differential equations for the modulus $\rho$ and the 
	argument $\phi$ of the order parameter \eqref{order-p} in the limit as $n\to\infty$. Using
	the combination of self-consistent analysis and numerical bifurcation techniques, both studies
	revealed an Andronov-Hopf, Bogdanov-Takens, and  homoclinic bifurcations in the parameter
	regime relevant to transition to synchronization. Similar results were obtained for
	a coupled system of noisy rotators in \cite{ZaksNeiman2003} albeit via a different approximation
	procedure.
	Previous studies reveal a rich bifurcation structure underlying macroscopic dynamics
	in the coupled model.
	In this paper, we focus on the origins of the slow-fast oscillations of the order parameters and,
	in particular, on the transition from the oscillatory to rotational motion of the order parameter. 
	We use the qualitative methods  for ordinary differential equations to  elucidate the bifurcation
	structure of the coupled model\footnote{Note that our setting is slightly different from those of
		models in
		\cite{KlinFran2015, TesSciCol2007, ZaksNeiman2003}.}.
	Specifically, we locate the Andronov-Hopf and Bogdanov-Takens
	bifurcations analitically and identify the relavant normal form for the Bogdanov-Takens
	bifurcation.
	Furthermore, we relate the transition to the rotational motion to a heteroclinic bifurcation
	of a family of vector fields on a cylinder.  {A related global bifurcation separating
		oscillations from rotations of  the order parameter in the model of noisy active
		rotators was reported
		(but not analyzed) in \cite{ZaksNeiman2003}. In Section~\ref{sec.heteroclinic},
		we provide a detailed analysis of the heteroclinic bifurcation for the model at hand.
	}

	There is an extensive literature on the KM. Early papers
	were devoted mostly to synchronization (cf.~\cite{Kur75, StrMir91, StrMir92, Str00}). The scope of
	more recent contributions encompasses  rigorous
	mathematical aspects of synchronization \cite{Chi15a, Die16, ChiMed19a},
	complex spatiotemporal patterns \cite{PanAbr15, Ome18}, generalizations or adaptations of
	the classical KM \cite{CMM18, ChiStr08, Laing09}, as well as
	applications to physical and biochemical systems  (see, e.g., \cite{ChiExp-1,ChiExp-2,ChiExp-3}).
	Our work is close in spirit to \cite{ChiStr08, WilStr06} where the modifications of the KM
	were used to tackle challenging questions about collective dynamics. Specifically, we wanted to understand
	collective dynamics in populations of type I excitable oscillators 
	with randomly distributed intrinsic frequencies.
	Coupled networks of this type
	have been studied in computational neuroscience (see, e.g.,
	\cite{ErmKop84, Erm96, MedZhu12}) albeit in
	settings that differ from that adopted here. In particular,
	methods for studying weakly coupled networks (cf. \cite{HopIzh97})
	do not apply  to systems
	with random parameters like \eqref{KM}.
	For coupled systems with different coupling type,
	collective dynamics in closely related models of coupled theta neurons
	was analyzed in \cite{Laing14, JHM21, OmeLai2022}.
	We believe that the analysis of the modified KM
	used in this paper
	complements the previous studies of coupled active
	rotators in \cite{KlinFran2015, TesSciCol2007, ZaksNeiman2003} and 
	brings new insights to understanding collective dynamics of type I oscillators. 
	\begin{figure}
		\centering
		{\bf a}\includegraphics[width=	.45\textwidth]{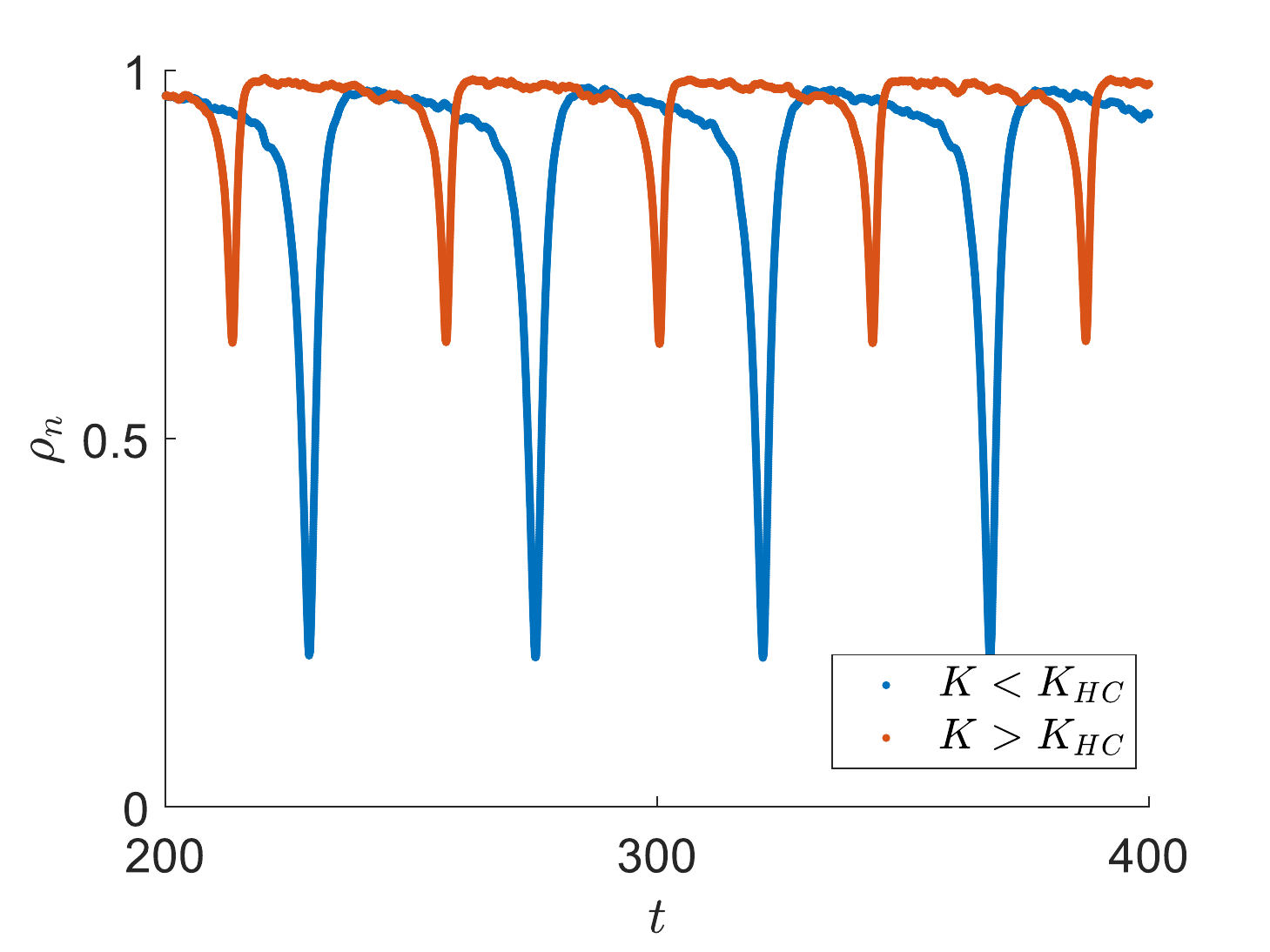}
		{\bf b}\includegraphics[width=	.45\textwidth]{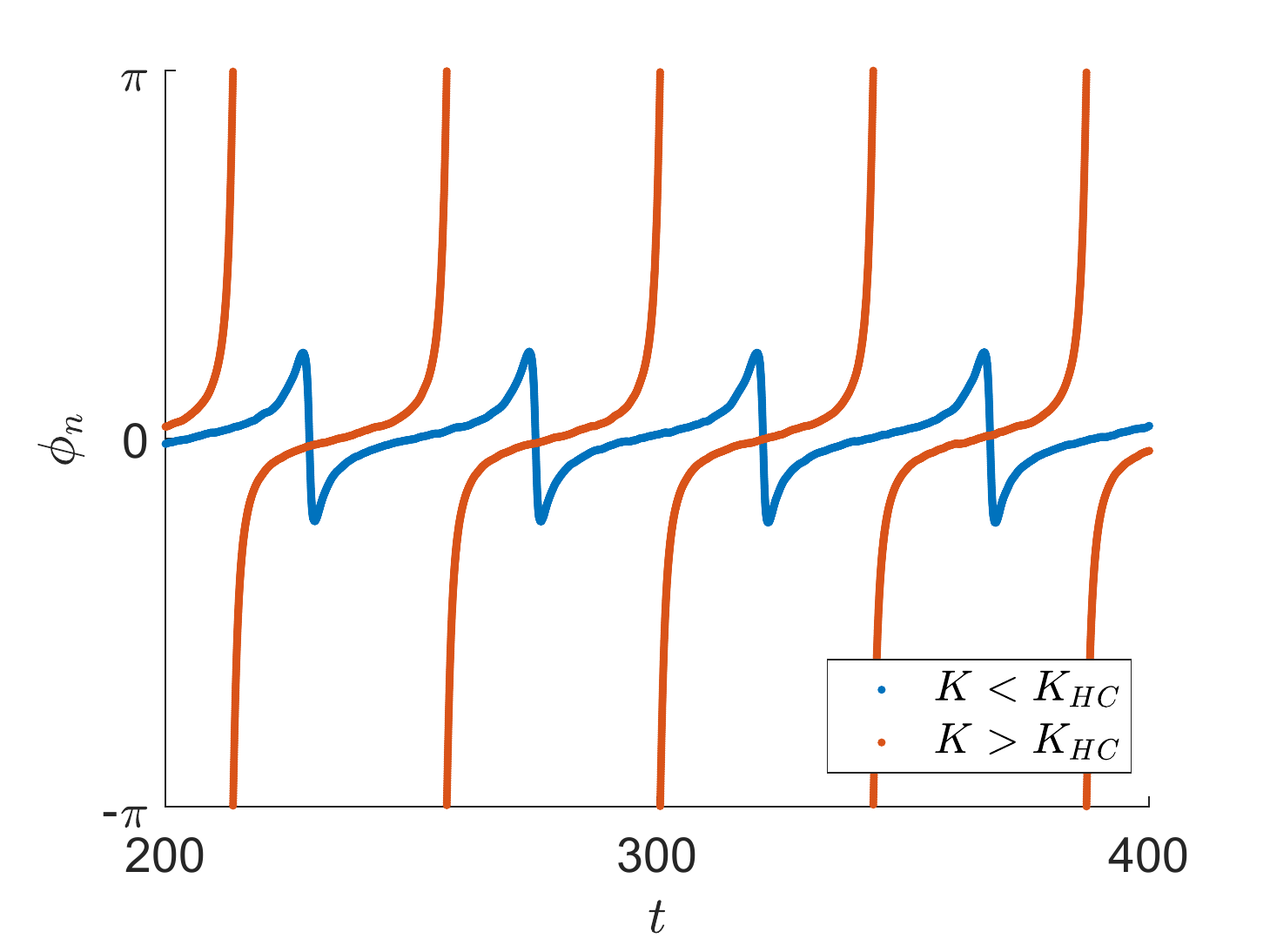}
		\caption{
			The time series of the modulus, $\rho_n$, and the argument, $\phi_n$,
			of the complex order parameter plotted for $K=0.23$ (blue) and $K=0.4$ (orange) with $n=2000$.
			Note that in \textbf{b} the oscillations of $|\phi_n|$ do not exceed $\pi/2$ for small $K$, whereas the oscillations
			of $\phi_n$ cover the entire range for $K$ sufficiently large. Intrinsic frequencies are sampled
			from the Lorentzian distribution \eqref{Lorentz} with $\epsilon = 0.1$, $\delta = 0.01$.
		} 
		\label{f.order}
	\end{figure}
	
	The paper is organized as follows. In the next section, we use the Ott-Antonsen Ansatz
	\cite{OttAnt08} to derive a system of two ordinary differential equations, which captures
	the long time dynamics of the coupled system. In Section~\ref{sec.qualitative}, we analyze
	the reduced system. The analysis uses the unfolding of the Bogdanov-Takens bifurcation
	among other qualitative techniques for ordinary differential
	equations. Further, we identify
	a heteroclinic bifurcation, which explains the
	phase transition in the collective dynamics. The heteroclinic bifurcation is analyzed in
	Section~\ref{sec.heteroclinic}. In Section~\ref{sec.collective}, we relate the
	analysis of the reduced system to the collective dynamics of \eqref{KM}. We conclude
	with a brief discussion of the main results in Section~\ref{sec.discuss}.
	
	\begin{figure}
		\begin{center}
			\textbf{a}\;     \includegraphics[width =.31\textwidth]{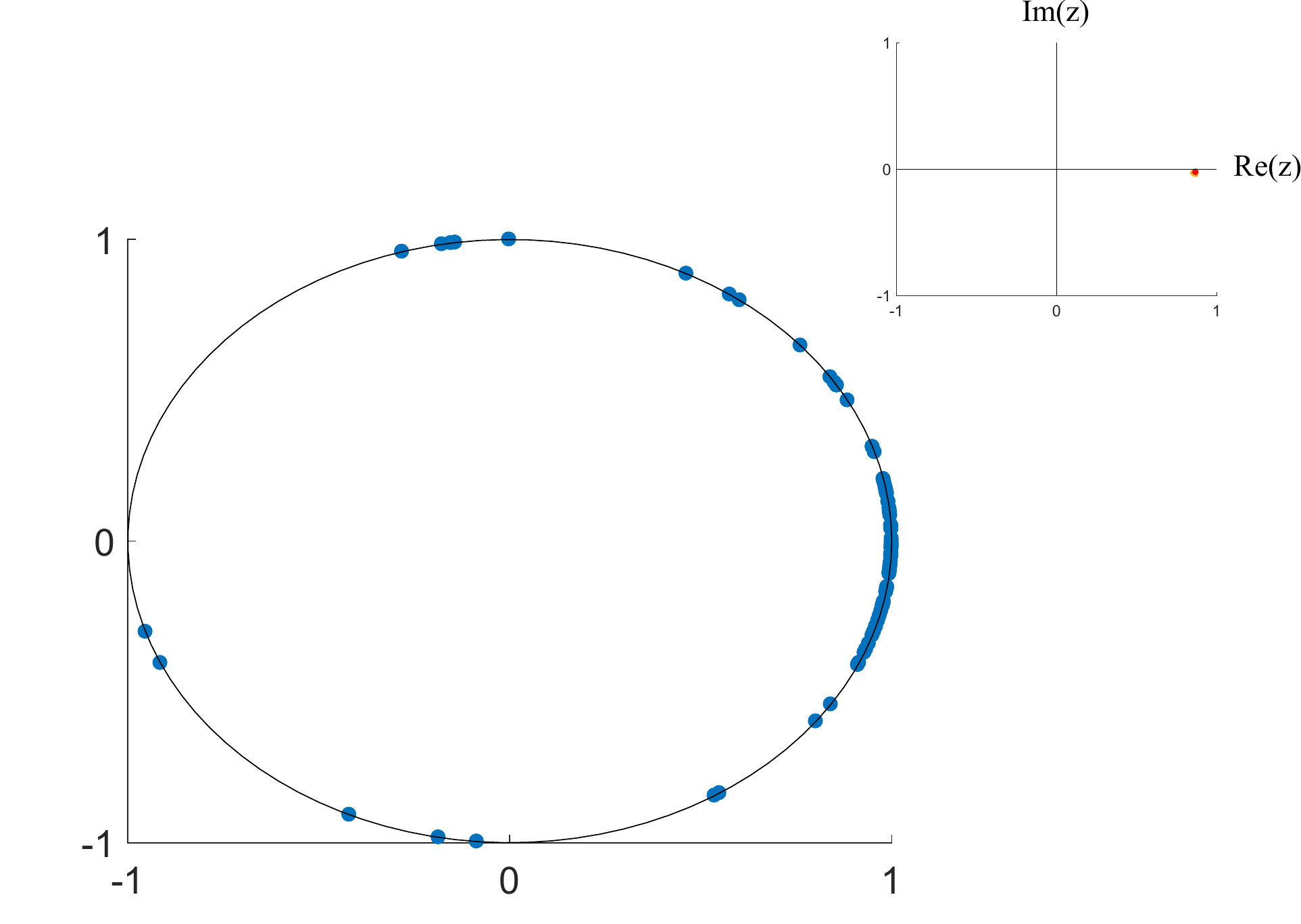}
			\textbf{b}\;     \includegraphics[width =.31\textwidth]{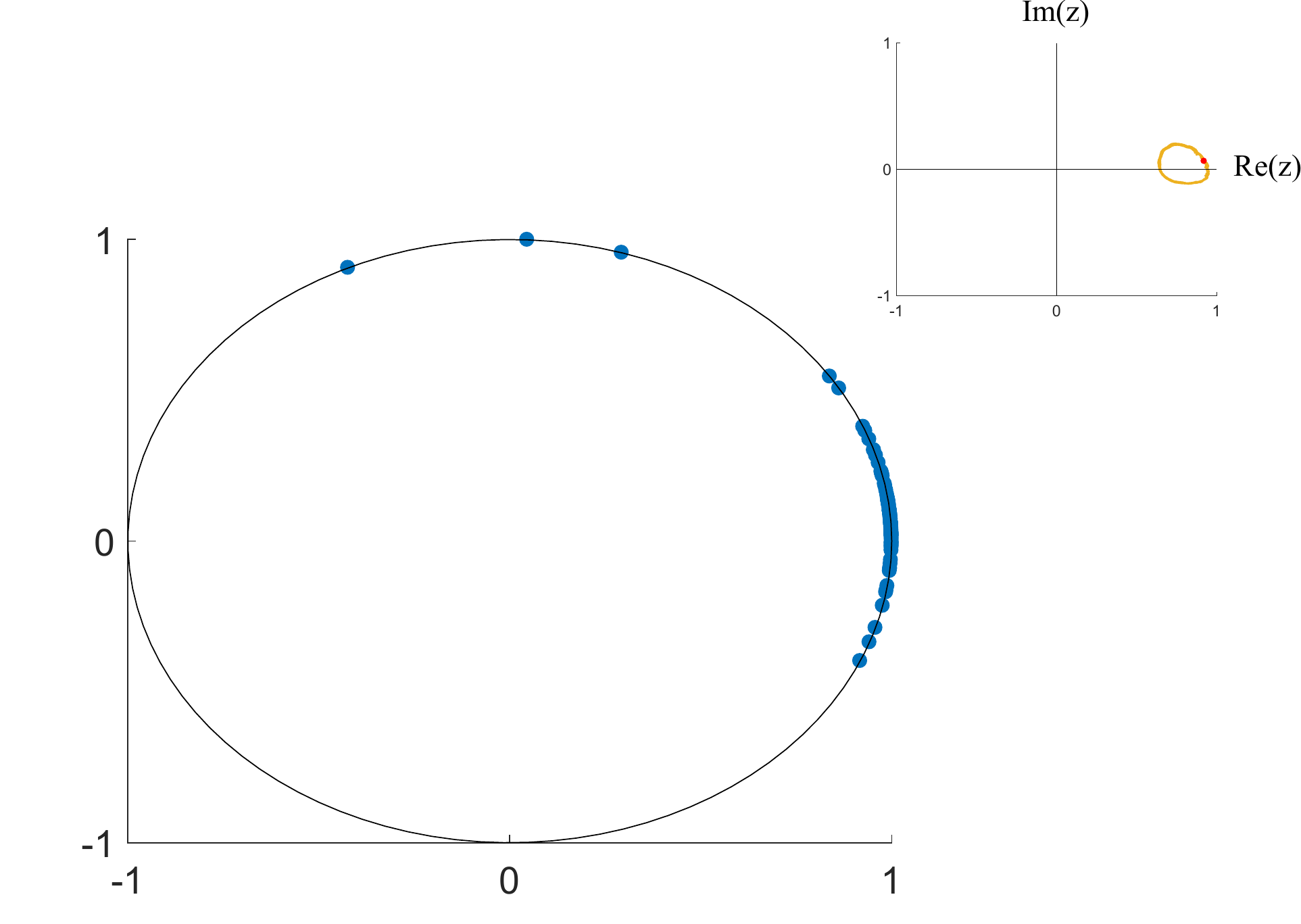}\\
			\textbf{c}\;	\includegraphics[width=.31\textwidth]{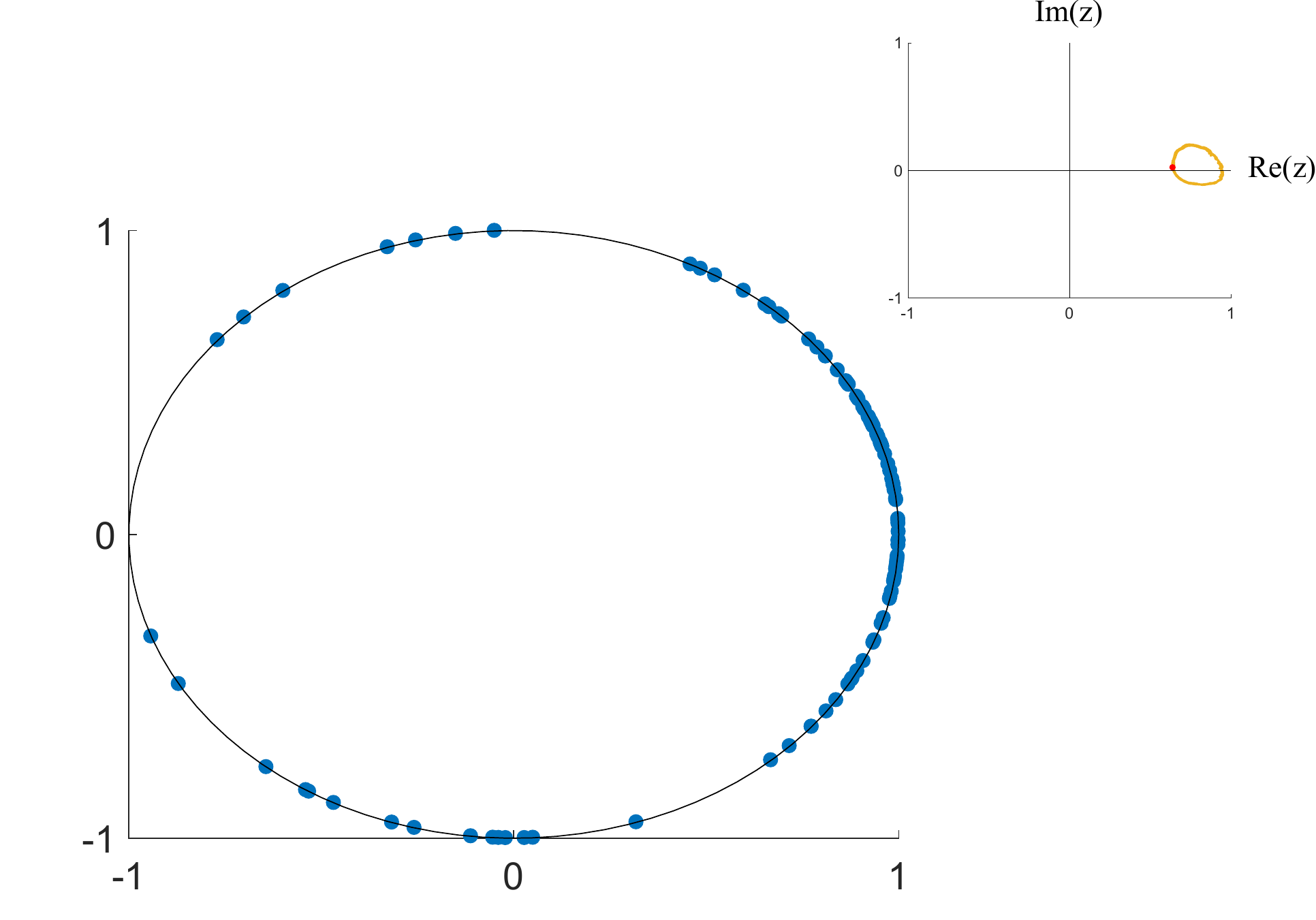}\\
			\textbf{d}\;      \includegraphics[width =.31\textwidth]{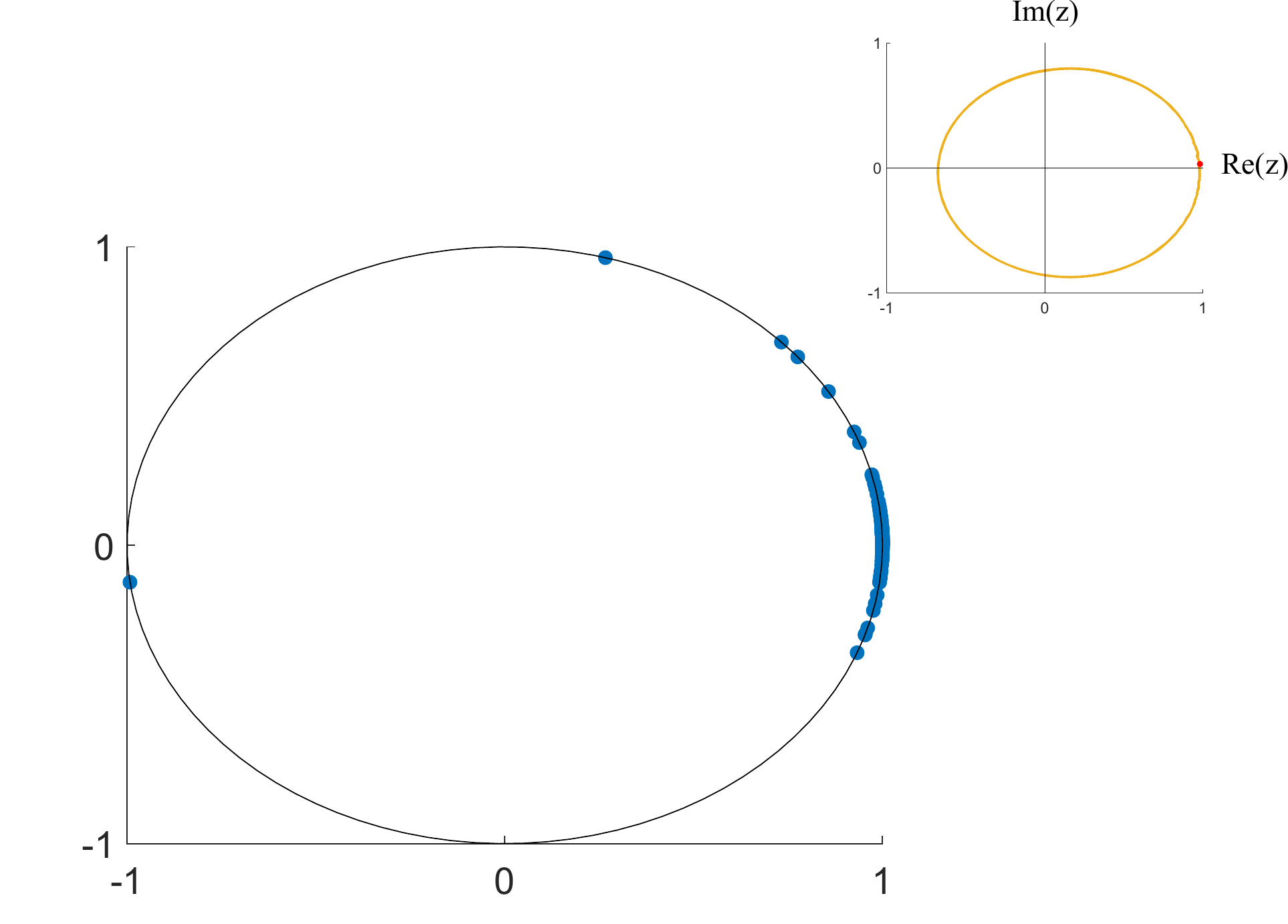}
			\textbf{e}\;	\includegraphics[width=.31\textwidth]{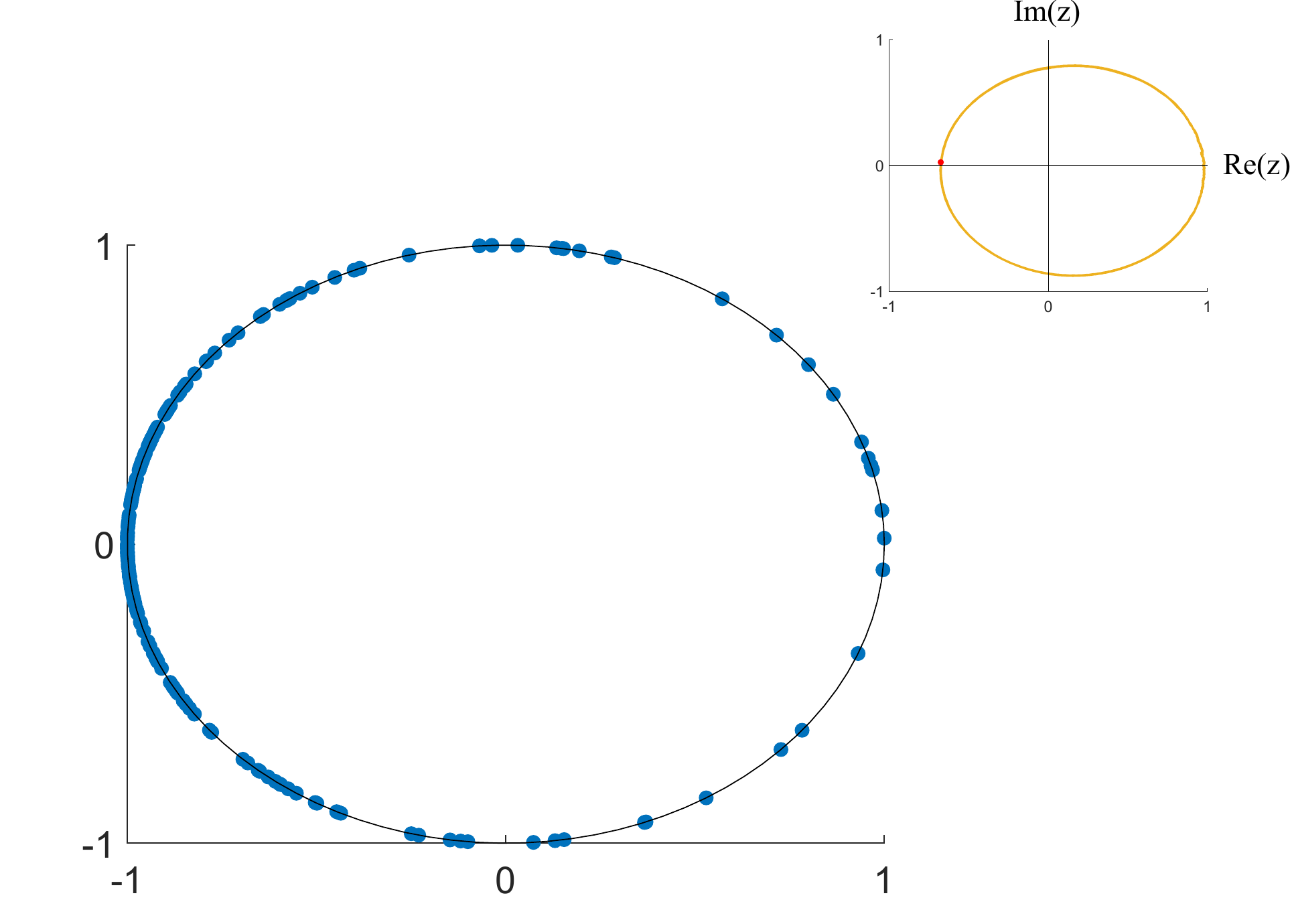}
		\end{center}
		\caption{Plots in each row illustrate three distinct regimes in collective
			dynamics for three values of $K$: ($K=0.05$ in \textbf{a}), ($K=0.17$ in \textbf{b},\textbf{c})
			and ($K=0.4$ in \textbf{d},\textbf{e}). The values of other parameters are
			$\delta = 0.01$, $\epsilon = 0.1$, and  $n=2000$ (though for ease of visualization we plot a random sampling of $200$ oscillators). Snapshot \textbf{a} shows the equilibrium distribution (though we emphasize that individual oscillators travel around the circle). The snapshots in \textbf{b}, \textbf{c}
			show the location of the oscillators in phase space at two moments of time during one period
			of oscillations. Note that the coherence varies from  high in \textbf{b} to low in \textbf{c}.
			The variation in coherence is also observed in the snapshots made for larger value of $K$ shown in \textbf{d},\textbf{e}.
			However, in the former case the oscillations of the center of mass of the population are restricted
			to the left half-plane, whereas in the latter case the center of mass travels around the circle.
			This can be seen from the trajectories of the complex order parameter shown in the insets.
			Note that the orbit traced by the order parameter in \textbf{d}, \textbf{e} encircles the origin.
			In contrast,  in \textbf{b}, \textbf{c} the argument of the order parameter remains between
			$(-\pi/2, \pi/2)$ for all times.
		}   
		\label{f.snapshots}
	\end{figure}
	
	\section{The mean field limit}\label{sec.mean-field}
	\setcounter{equation}{0}
	In the large $n$ limit, the dynamics of \eqref{KM} is described by the following Vlasov equation
	(cf.~\cite{Dob79, StrMir91})
	\begin{equation}\lbl{MF}
		\p_t f(t,\theta,\omega) +
		\p_\theta \left\{ V(t,\theta,\omega) f(t,\theta,\omega) \right\}=0,
	\end{equation}
	where  $f(t,\theta,\omega)d\theta d\omega$ is the  probability of
	$(\theta,\omega)\in (\theta, \theta+d\theta)\times (\omega, \omega+d\omega)$
	at time $t\in \R^+$.
	The velocity field is defined by
	\be\lbl{def-V}
	V(t,\theta,\omega)=1+\omega-\frac{e^{\1\theta}+e^{-\1\theta}}{2}
	+ {K\over 2\1}  \left( h(t) e^{-\1\theta} -\overline{h(t) e^{-\1\theta}}\right),
	\ee
	where 
	\be\lbl{order}
	h(t):= h[f(t,\cdot)]=\int_{\T\times\R}  e^{\1\theta}  f(t,\theta,\omega)  d\theta d\omega,
	\ee
	is the continuum order parameter.
	
	From \eqref{KM} with $K=0$ one can determine the stationary distribution of the
	oscillators in the phase space:
	$$
	f_{st}(\theta, \omega)=\frac{j(\omega)}{1+\omega-\cos\theta},
	$$
	where $j(\omega)$ is subject to the following equation
	$$
	j(\omega) \int_{\T}\frac{d\theta} {1+\omega-\cos\theta} = g(\omega).
	$$
	
	To investigate \eqref{MF} for $K\ge 0$, we invoke the Ott-Antonsen Ansatz
	\cite{OttAnt08}
	\begin{equation}\lbl{OA}
		f(t,\theta,\omega)=\frac{g(\omega)}{2\pi} \left\{ 1+
		\sum_{k=1}^\infty
		\left[ \alpha(t,\omega)^k e^{\1k\theta}+ \overline{\alpha(t,\omega)}^k e^{-\1k\theta} \right] \right\}.
	\end{equation}
	Previous studies of the KM suggest that the Ott-Antonsen Ansatz correctly
	describes the long time asymptotic behavior of solutions of this model. It has been used
	to study chimera states \cite{Ome18} among many other spatiotemporal patterns
	in the KM and related models \cite{OttAnt09, Laing09, KlinFran2015, TesSciCol2007, OmeLai2022}.
	There is a convincing albeit not completely rigorous mathematical argument justifying
	the use of this Ansatz for studying the long time behavior of solutions of the Kuramoto
	model. The analysis in the remainder of this paper relies on the validity of the Ott-Antonsen
	Ansatz.

	By plugging \eqref{OA} into \eqref{MF}, we obtain
	\begin{equation}\label{prealpha}
		\p_t \alpha =\frac{K}{2}\left(\bar h-h\alpha^2\right)+\1\left(\frac{1}{2}-\alpha\omega +\frac{\alpha^2}{2}\right).
	\end{equation}
	Multiplying both sides by $2$ and rescaling time in \eqref{prealpha} yields
	\begin{equation}\label{alpha-pde}
		\p_t \alpha =K\left(\bar h-h\alpha^2\right)+\1\left(1-2\alpha\omega +\alpha^2\right).
	\end{equation}
	Equation \eqref{alpha-pde} is simpler than the Vlasov equation \eqref{MF}, but it is still an integro-partial
	differential equation. For the Lorentzian density $g(\omega)$ given in \eqref{Lorentz}, one
	can further reduce \eqref{alpha-pde}
	to the following ordinary differential equation (cf.~\cite{OttAnt08})
	\begin{equation}\label{alpha}
		\alpha^\prime = -2\alpha \delta +K\left( \alpha-\alpha|\alpha|^2\right) +
		\1\left( 1-2\alpha\left(1+\epsilon^2\right) +\alpha^2\right),
	\end{equation}
	where by abuse of notation we continue to denote $\alpha(t, \epsilon^2-\1\delta)$ by $\alpha(t)$. Furthermore,
	the order parameter $h$ can now be expressed through $\alpha$:
	\begin{equation}\label{new-h}
		h(t)=\overline{\alpha(t)}.
	\end{equation}

	Equation \eqref{alpha} can be written as an ordinary differential
	equation on $\R^2$:
	\begin{equation}\label{cart}
		\begin{pmatrix}
			x^\prime\\ y^\prime
		\end{pmatrix}
		= \begin{pmatrix} 0 \\ 1 \end{pmatrix}
		+  \begin{pmatrix} K-2\delta & 2(1+\epsilon^2)\\ -2(1+\epsilon) &
			K-2\delta \end{pmatrix}  \begin{pmatrix} x \\ y\end{pmatrix}
		+(x^2+y^2)  \begin{pmatrix} x \\ -1-y\end{pmatrix},
	\end{equation}
	where $\alpha=x+\1 y$.
	Since we are interested in the macroscopic dynamics of \eqref{KM},
	we want to track the changes the qualitative behavior of the modulus and the argument
	of the order parameter. To this end, we rewrite \eqref{alpha} in polar coordinates
	$\alpha=\rho e^{\1\phi}$:
	\begin{equation}\label{2odes}
		\begin{split}
			\rho^\prime &= \left(1-\rho^2\right) \left(\sin\phi +K\rho\right) -2\delta \rho,\\
			\phi^\prime &= \rho^{-1}\left(1+\rho^2\right) \cos\phi -2\left(1+\epsilon^2\right).
		\end{split}
	\end{equation}
	The polar coordinate transformation blows up the origin in the $x-y$ plane to a $\rho=0$ circle.
	This singular change preserves a
	$1-1$ correspondence between the trajectories of \eqref{cart} and \eqref{2odes} lying in
	$\R^2/\{0\}$ and $R^+\times\T$
	respectively (cf.~\cite[\S 2.8]{ArrowPlace}). The two systems are topologically equivalent in
	these domains. 
	The right-hand side of \eqref{2odes} has a singularity at $\rho=0$. To resolve this
	singularity, we multiply both sides of the equations in \eqref{2odes} by $\rho$ and rescale time
	to obtain
	\begin{equation}\label{system}
		\begin{split}
			\dot \rho &= \rho\left(1-\rho^2\right) \left(\sin\phi +K\rho\right) -2\delta \rho^2,\\
			\dot \phi &= \left(1+\rho^2\right) \cos\phi -2\left(1+\epsilon^2\right)\rho.
		\end{split}
	\end{equation}
	Systems \eqref{2odes} and \eqref{system} are topologically equivalent on $\R^+\times \T$.
	However, the latter defines a smooth vector field over $\R\times \T$. It can be studied
	by standard methods of the qualitative theory for differential equations. In particular,
	this allows us to resolve the singularity of \eqref{2odes} at $\rho=0$, which is important
	for understanding the macroscopic dynamics of the coupled system.

	\section{Qualitative analysis of the planar system}
	\label{sec.qualitative}
	\setcounter{equation}{0}
	In this section, we analyze \eqref{system} when $K,\delta,\epsilon$ are nonnegative and small.
	
	There are two fixed points of \eqref{system}:
	\begin{equation}\label{fixed}
		\left(0,\frac{\pi}{2}\right), \left(0,\frac{-\pi}{2}\right).
	\end{equation}
	Additionally, when $\delta = \epsilon = 0$ we can explicitly calculate a third fixed point at $(1,0)$.
	We begin with the local analysis around $\left(1, 0\right)$. This is followed by the linearization about the other two fixed points. After that we combine this information and identify a nonlocal
	bifurcation of  the vector field responsible for the transformation of the collective dynamics in \eqref{KM}.
	
	\subsection{The Bogdanov-Takens singularity}
	Let $\rho=1-r$ and rewrite \eqref{system} in the neighborhood of $(1, 0)$ keeping terms up to
	and including order $2$ to obtain
	\begin{equation}\label{BT-delta-epsilon}
		\begin{split}
			\dot r &= -2r\phi -2r K +2\delta + O(3),\\
			\dot \phi &= -\phi^2 + r^2 -2\epsilon^2 +O(3).
		\end{split}
	\end{equation}
	With $\epsilon=\delta=K=0$ \eqref{BT-delta-epsilon} fits into the normal form of the Bogdanov-Takens
	bifurcation \cite{Arnold72, Bogdanov75, Takens74}:
	\begin{equation}\label{BT}
		\begin{split}
			\dot r &= -2r\phi,\\
			\dot \phi &= -\phi^2 + r^2.
		\end{split}
	\end{equation}
	Rewriting \eqref{BT} in polar coordinates
	$r=a\cos\psi, \; \phi=a \sin\psi$, we obtain
	\begin{equation}\label{BT-polar}
		\begin{split}
			\dot a  &=-a^2\sin\psi,\\
			\dot\psi &= a\cos\psi.
		\end{split}
	\end{equation}
	From   \eqref{BT-polar} one can plot the phase portrait of \eqref{BT} (see Figure~\ref{f.BT}\textbf{a}).
	\begin{figure}
		\centering
		\textbf{a}\,	\includegraphics[width=.4\textwidth]{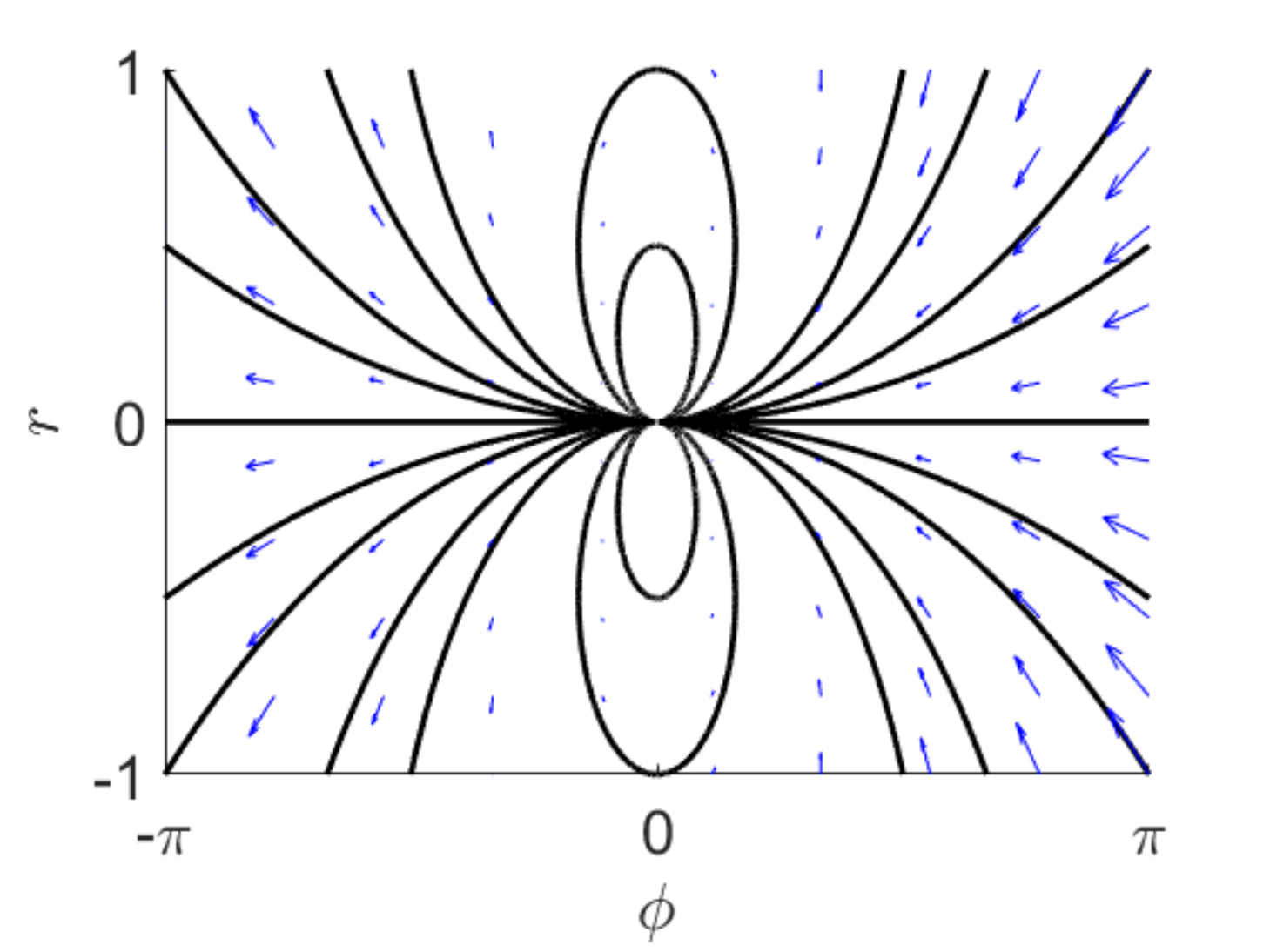}\quad
		\textbf{b}\,	\includegraphics[width=.4\textwidth]{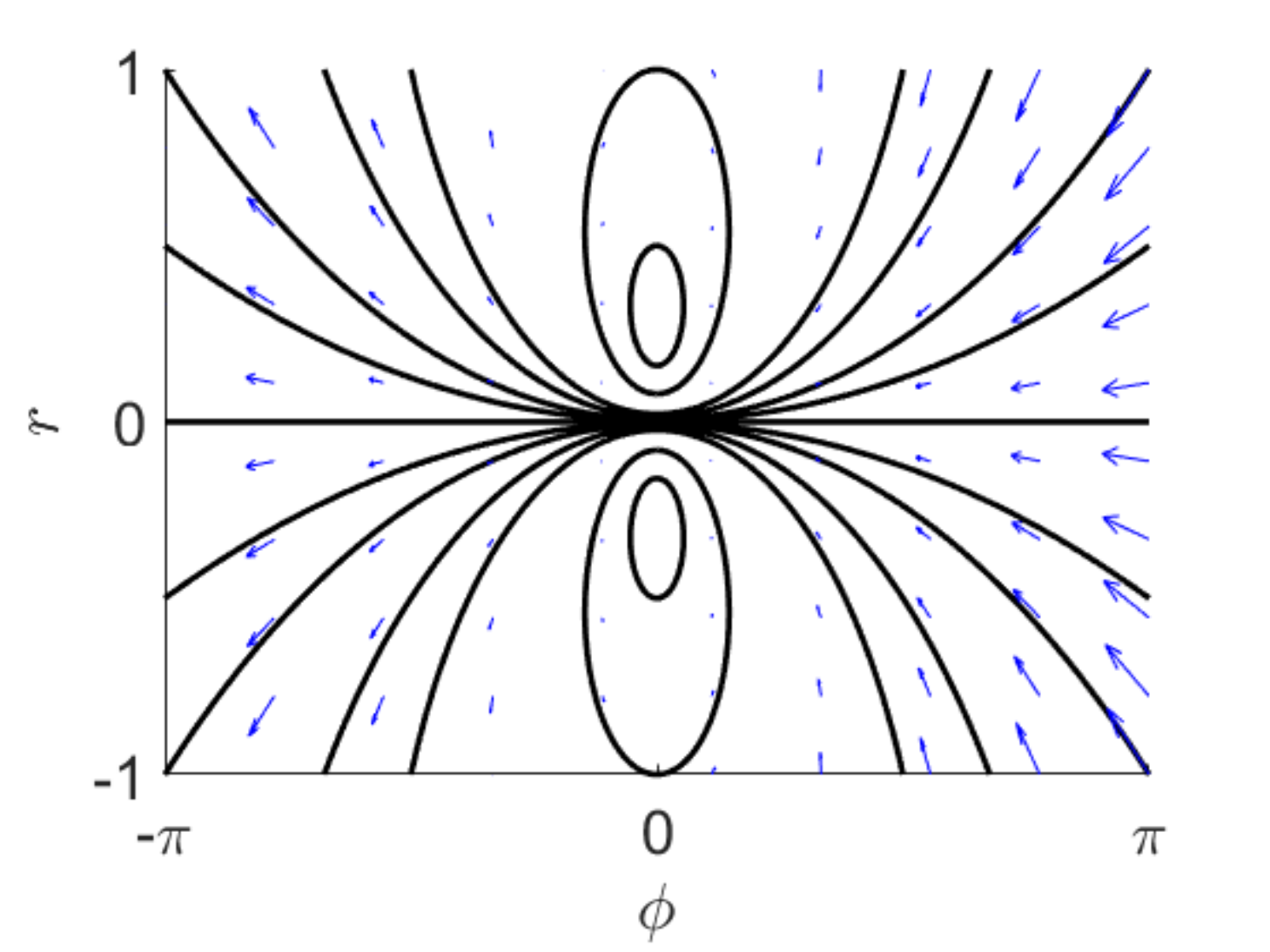}\\
		\textbf{c}\,      \includegraphics[width=.4\textwidth]{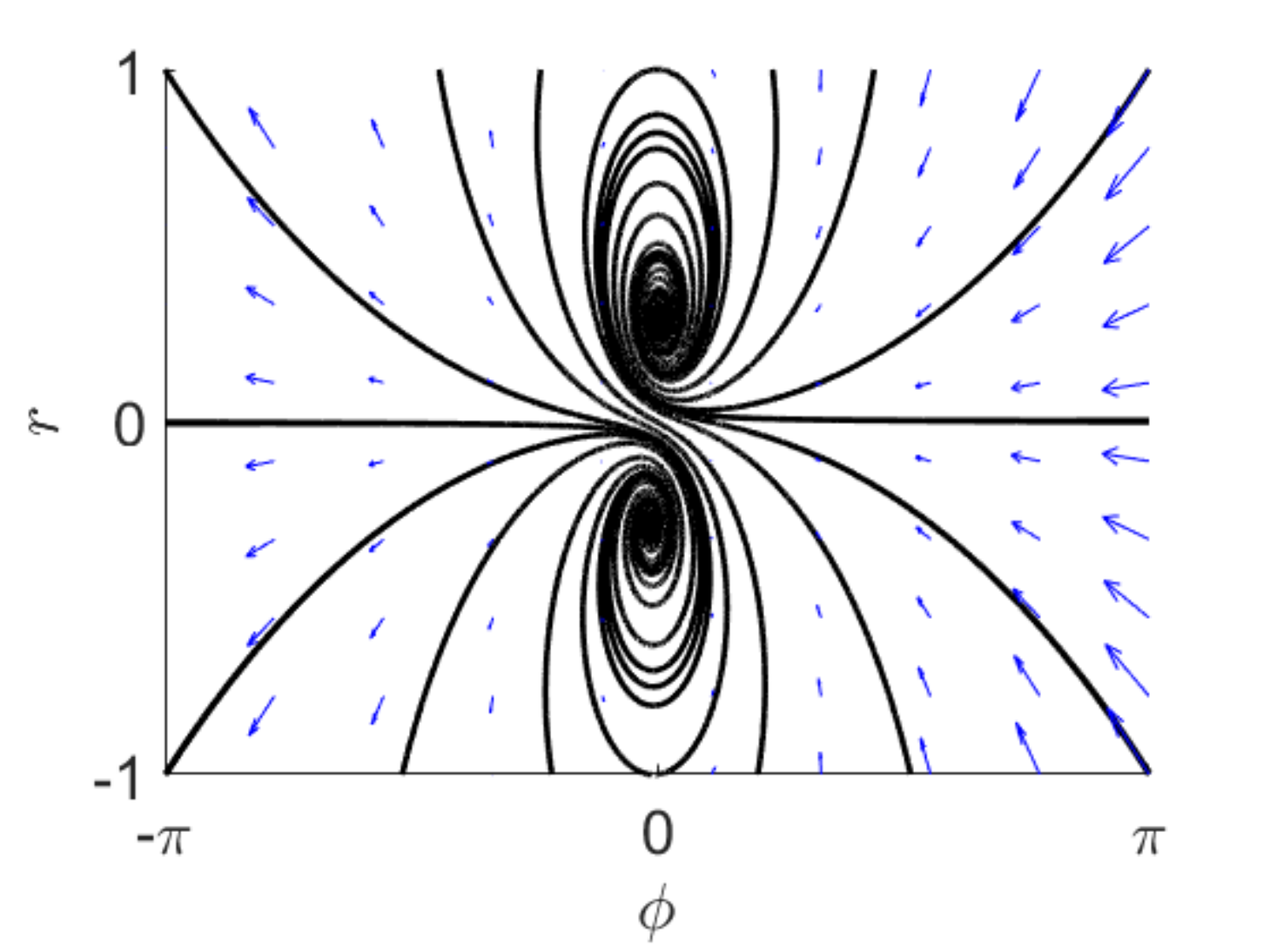}\quad
		\textbf{d}\,	\includegraphics[width=.4\textwidth]{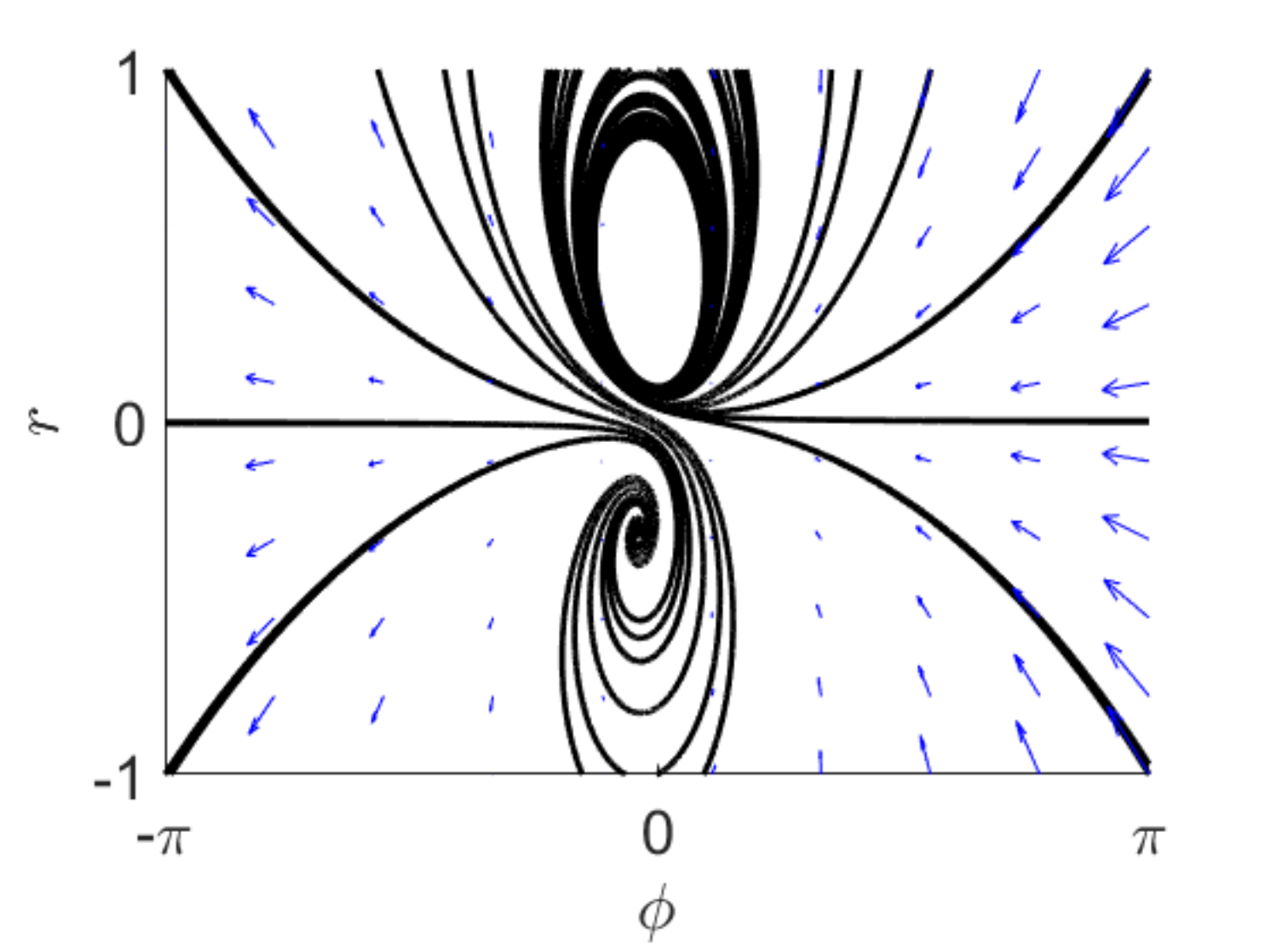}
		\caption{Phase portraits for \eqref{BT-delta-epsilon}.
			{\bf a)} For $K=\delta = \epsilon =0$, the system is at the Bogdanov-Takens bifurcation.
			{\bf b)} Upon increasing $\epsilon$, it becomes an integrable system with two centers
			(see \eqref{integral}; $K=\delta = 0$, $\epsilon=0.2$).
			{\bf c)} Further, increasing  $K$ and $\delta$ transforms the two centers into foci
			($K=\delta = 0.01$, $\epsilon = 0.2$).
			{\bf d)} For larger $K$, the system undergoes an Andronov-Hopf bifurcation, which  gives
			birth to a stable limit cycle ($K= 0.09$, $\delta = 0.01$, $\epsilon = 0.2$).}
		\label{f.BT}
	\end{figure}

	Next, we keep $\delta=K=0$ and let $0<\epsilon\ll 1$. By increasing $\epsilon$ from $0$,
	we observe that the
	fixed point at the origin
	bifurcates into two fixed points: $(0, \sqrt{2}\epsilon)$ and $(0,-\sqrt{2}\epsilon)$.
	Furthermore, the system has
	an integral (cf.~\cite{GuckHolmes})
	\begin{equation}\label{integral}
		J(r,\phi)=\frac{r^2+\phi^2 +\epsilon^2}{r}.
	\end{equation}
	Thus, the trajectories are given by two families of circles:
	$$
	(r-c)^2+\phi^2=c^2-\epsilon^2
	$$
	for $c\ge \epsilon$ and $c\le -\epsilon$. As $c^2\to\infty$ they limit onto $r$-axis
	(Fig.~\ref{f.BT}\textbf{b}).
	
	\begin{figure}
		\centering
		\textbf{a}\,	\includegraphics[width=.3\textwidth]{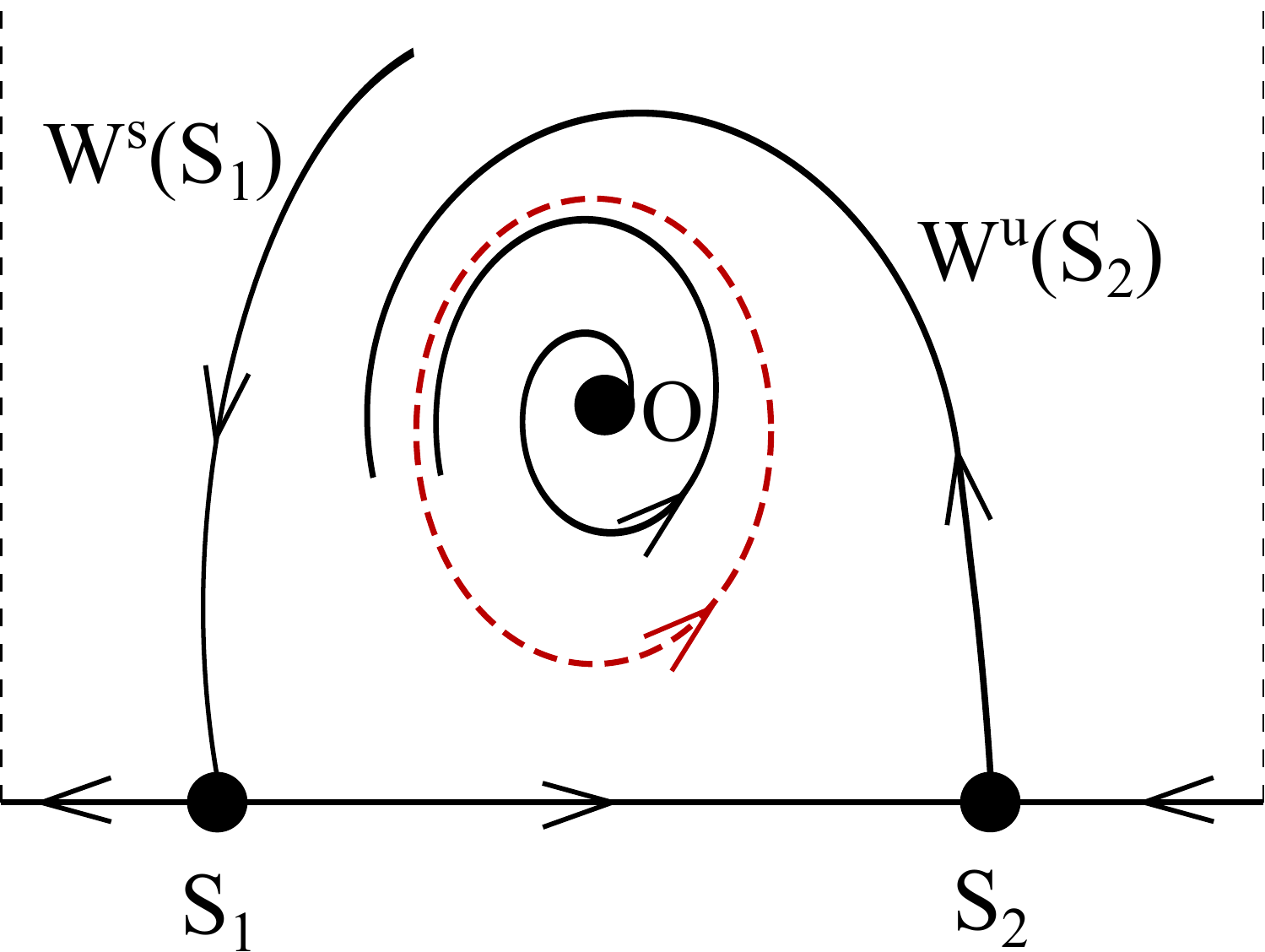}
		\textbf{b}\,	\includegraphics[width=.3\textwidth]{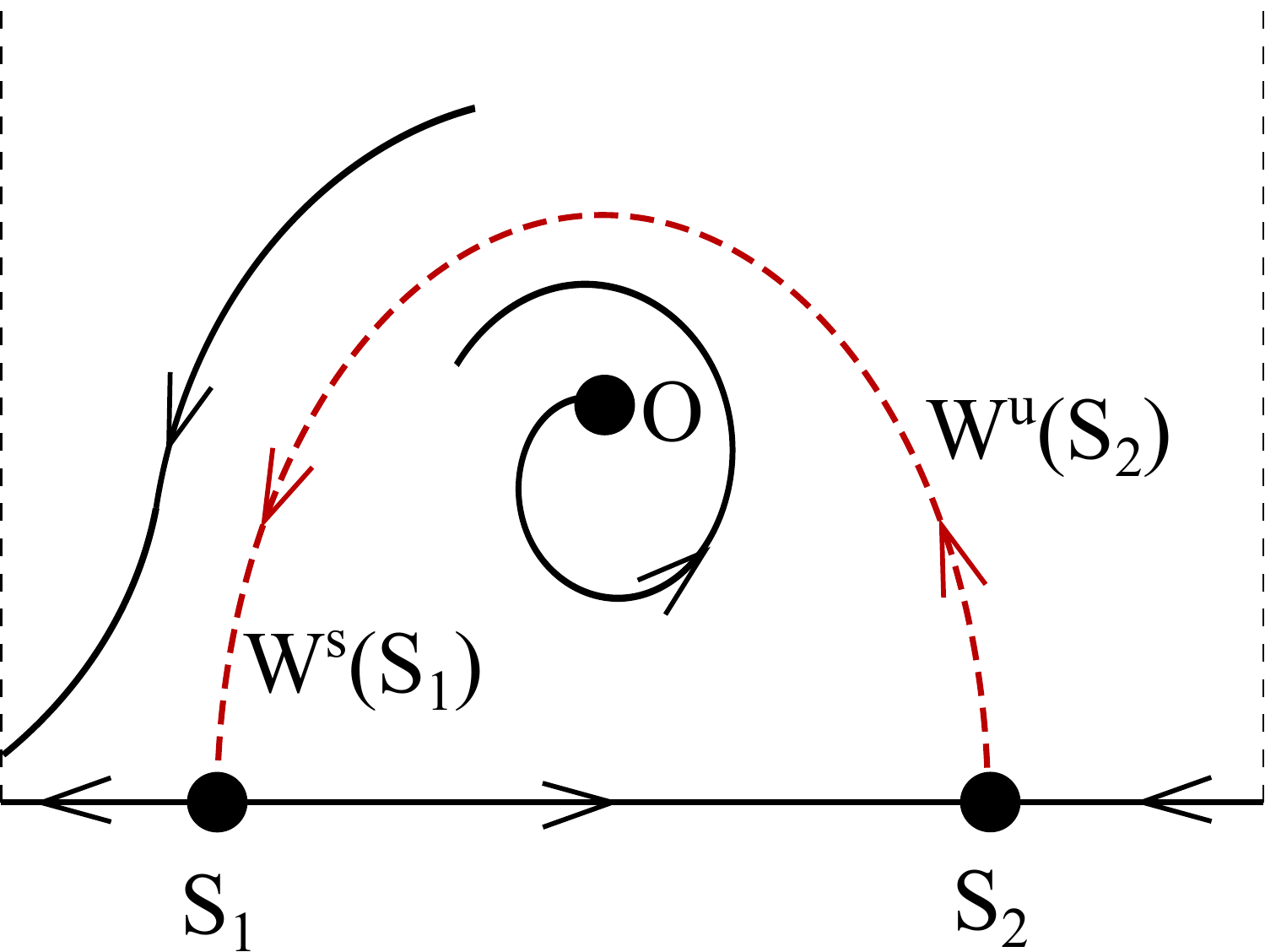}
		\textbf{c}\,      \includegraphics[width=.3\textwidth]{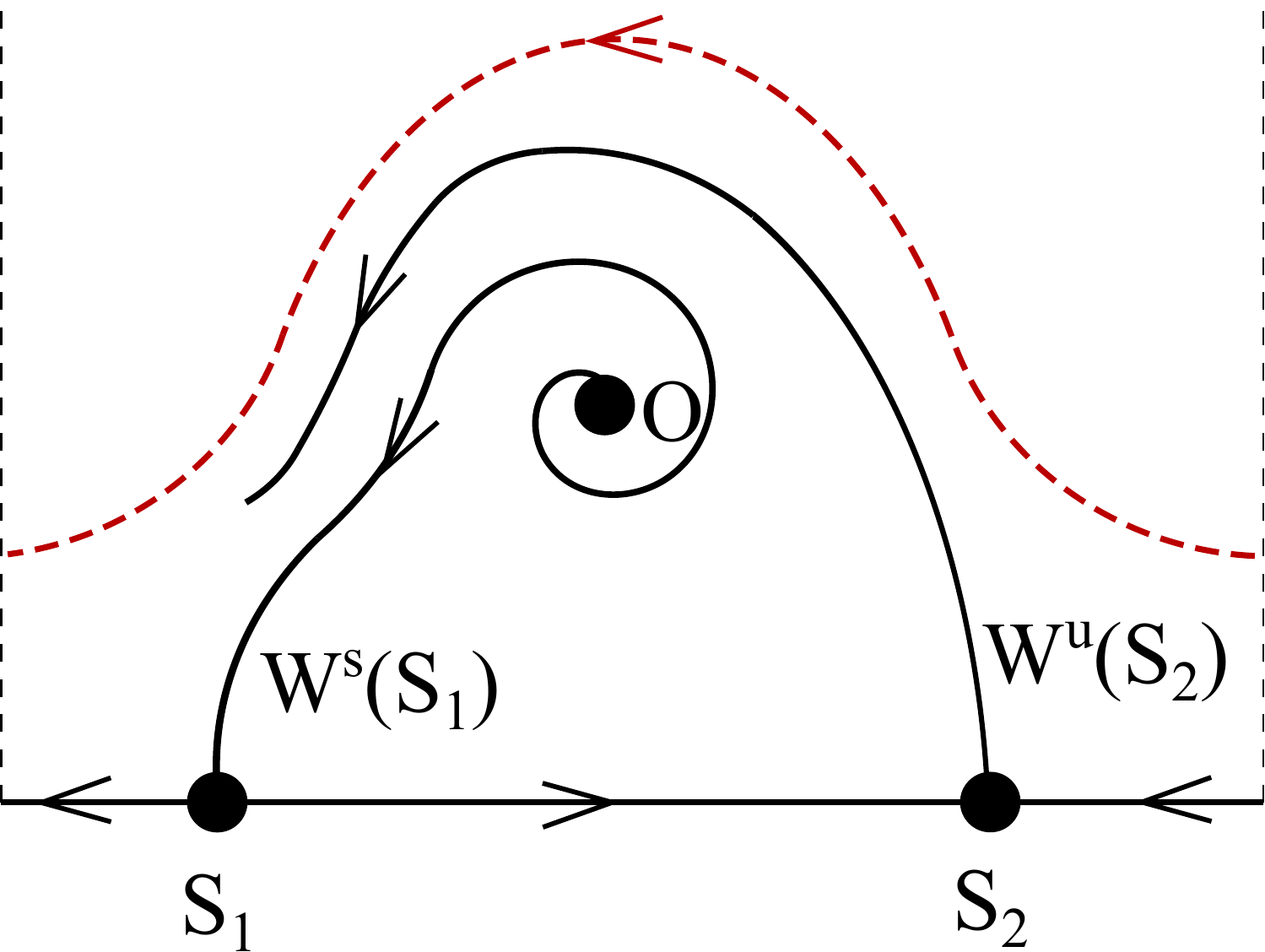}
		\caption{Schematic phase portraits just before, at, and right after
			the  heteroclinic bifurcation.  The periodic orbits in \textbf{a} and \textbf{c}
			(plotted in red) bifurcate from the heteroclinic loops shown  in \textbf{b}.
			Here and below, the phase portraits are identified along lateral sides (see
			dashed lines).
		}
		\label{f.heteroclinic}
	\end{figure}

	Next suppose $\epsilon, \delta,$ and $K$ are positive and small. To study this case we use
	the following scaling:
	\begin{equation}\label{scaling}
		\delta =\nu\epsilon, K=\mu\epsilon, r:=\epsilon\tilde r, \phi=\epsilon\tilde \phi.
	\end{equation}
	By plugging \eqref{scaling} into \eqref{BT-delta-epsilon}, we have
	\begin{equation}\label{r-phi}
		\begin{split}
			\epsilon^{-1} \dot{\tilde r} &= -2\tilde r \tilde\phi -2\mu \tilde r +2 \nu,\\
			\epsilon^{-1} \dot{\tilde\phi} &= -\tilde{\phi}^2 + \tilde{r}^2 -1.
		\end{split}
	\end{equation}
	Treating $\nu$ and $\mu$ as positive and small parameters of the same order, we locate the 
	fixed point of \eqref{r-phi} $(\bar r, \bar\phi)$ with positive $\bar r$:
	$$
	\bar \phi\approx \nu-\mu, \bar r\approx 1.
	$$
	Linearization about $(\bar r, \bar\phi)$ yields:
	$$
	\epsilon^{-1}\dot\xi=
	A \xi,
	\quad
	A= -2 \begin{pmatrix} \bar\phi+\mu & \bar r \\  -\bar r & \bar\phi \end{pmatrix}.
	$$
	Note 
	$$
	\operatorname{Tr} A =-4\nu+2\mu, \quad \det A=4\left( \bar\phi^2 +\bar r^2 +\bar\phi \mu\right).
	$$
	The determinant of $A$ is positive for sufficiently small $\mu\ge 0$ (recall that $\bar r=O(1)$).
	For $0\le \mu<2\nu$ the fixed point $(\bar r,\bar\phi)$ is a stable focus. It becomes unstable
	for $\mu>2\nu$. At $\mu=2\nu$ the system undergoes a supercritical Andronov-Hopf bifurcation
	(see Figure~\ref{f.BT}\textbf{d}).
	

	\subsection{The heteroclinic bifurcation}
	We now turn to the remaining two fixed points: $S_1=(0,-\pi/2)$ and
	$S_2=(0,\pi/2)$. Linearization about $S_1$ yields
	$$
	\begin{pmatrix}\dot\xi \\ \dot \eta\end{pmatrix}=
	\begin{pmatrix} -1 & 0 \\ -2(1+\epsilon^2) & 1\end{pmatrix} \begin{pmatrix} \xi \\ \eta\end{pmatrix}
	+
	\begin{pmatrix} (K-2\delta)\xi^2 +O(3) \\  O(3) \end{pmatrix}.
	$$
	The eigenvalues are $\lambda_1=1$ and $\lambda_2=-1$
	with the corresponding eigenvectors $v_1=(1,0)$ and $v_2=(1, 1+\epsilon^2)$.
	
	Similarly,  linearization about $S_2$ yields
	$$
	\begin{pmatrix}\dot\xi \\ \dot \eta\end{pmatrix}=
	\begin{pmatrix} 1 & 0 \\ -2(1+\epsilon^2) & -1\end{pmatrix} \begin{pmatrix} \xi \\ \eta\end{pmatrix}
	+
	\begin{pmatrix} (K-2\delta)\xi^2 +O(3) \\  O(3) \end{pmatrix}.
	$$
	The eigenvalues are $\lambda_1=1$ and $\lambda_2=-1$
	with the corresponding eigenvectors $v_1= (1, -1-\epsilon^2) $ and $v_2=(1, 0)$.
	Consider $D=[-1, 1]\times \T$. In the parameter regime of interest, $D$ is positively invariant.
	Denote the branches of stable and unstable manifolds of $S_1$ and $S_2$ lying in $D$ by
	$W^s(S_{1,2})$ and $W^u(S_{1,2})$ respectively.

	In the remainder of this section, we assume that $0<\epsilon\ll 1$ is fixed.
	Both $K$ and $\delta $ are also nonnegative and $O(\epsilon)$. Further, we fix $\delta$ and treat
	$K$ as a control parameter.
	The linearization about $S_1$ and $S_2$ shows that both are saddles. The unstable manifold
	of $S_1$ and the stable manifold of $S_2$ lie in $\{\rho=0\}$ (Fig.~\ref{f.heteroclinic}).
	The unstable manifold
	of $S_2$ and the stable manifold of $S_1$ are tangent to $(1, -1-\epsilon^2)$ and
	$(1, 1+\epsilon^2)$ respectively.
	Recall that there is another fixed point in $D$:
	$O\approx (1-\epsilon, \delta-K)$. For $K>2\delta$, $O$ is an unstable focus. Denote the stable limit
	cycle born at the Andronov-Hopf bifurcation by $\mathcal{P}^-_K$. Note
	that $W^u(O)$ and $W^u(S_2)$ limit onto $\mathcal{P}^-_K$   (Fig.~\ref{f.heteroclinic}\textbf{a}).

	As $K$ is increased from the Andronov-Hopf bifurcation, $K_{AH}$, the following transformations of the phase portrait take place.
	The periodic orbit $\mathcal{P}^-_K$  grows in size. $W^u(S_2)$ moves up while $W^s(S_1)$   moves
	down (Figure~\ref{f.heteroclinic}\textbf{a}). At $K=K_{HC}$ they intersect forming a heteroclinic orbit
	$\Gamma^0$
	connecting $S_2$ and $S_1$ (Figure~\ref{f.heteroclinic}\textbf{b}).
	After the heteroclinic bifurcation, the limit cycle born at the Andronov-Hopf bifurcation disappears
	blending into a heteroclinc loop.
	A new limit cycle, $\mathcal{P}^+_K$, is born at $K=K_{HC}+0$ (Figure~\ref{f.heteroclinic}\textbf{c}).
	In contrast to  $\mathcal{P}^-_K$,  $\mathcal{P}^+_K$ is noncontractible.
	Thus, the heteroclinic bifurcation produces a limit cycle on each side of the bifurcation, i.e., both
	at $K=K_{HC}-0$ and $K=K_{HC}+0$. The bifurcation diagram in
	Figure~\ref{f.bif-diagram} summarizes this information.
	\begin{figure}
		\centering
		\includegraphics[width=.45\textwidth]{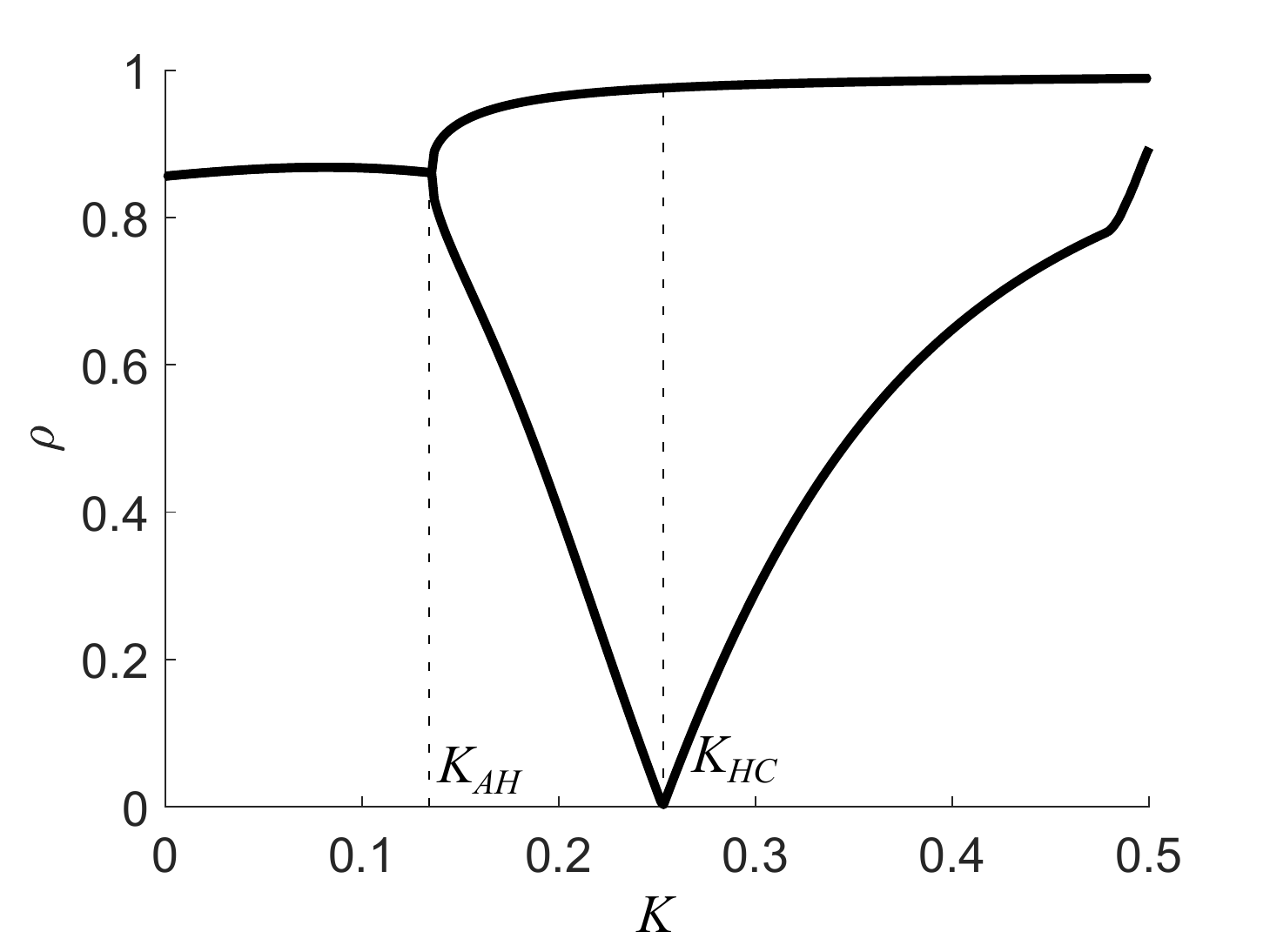}
		\caption{The bifurcation diagram for \eqref{BT} with
			$0<\delta,\epsilon\ll 1$ (here, $\epsilon = 0.1, \delta = 0.01$). For families of periodic orbit, the
			largest and the smallest values of $\rho$ are indicated.
		}
		\label{f.bif-diagram}
	\end{figure}

	To get a first insight into the heteroclinic bifurcation, we numerically computed
	a Poincar\'e map for \eqref{system}. Specifically, for values of $K$ near $K_{HC}$, we set up
	a cross section, $\Sigma$, intersecting
	the heteroclinic orbit (see a dashed line in plots \textbf{a} and \textbf{b} in Figure~\ref{f.Pmap}).
	Then we sampled initial conditions from the appropriate region of $\Sigma$ (around the point
	of intersection with $\Gamma^0$) and followed the corresponding trajectories until their first
	return to the chosen region of $\Sigma$ (see trajectories plotted in red
	of Figure~\ref{f.Pmap}\textbf{a},\textbf{b}). Note that after the bifurcation the red trajectory
	makes a full revolution around the cylinder before returning back to $\Sigma$
	(see Figure~\ref{f.Pmap}\textbf{b}). The Poincar\'e maps computed before and after
	the bifurcation are practically identical (see Figure~\ref{f.Pmap}\textbf{c},\textbf{d}).
	Furthermore, the map shows a strong contraction of the vector field near the heteroclinic loops.
	The periodic orbits corresponding to the fixed points of the map shown in
	Figure~\ref{f.Pmap}\textbf{c} and \textbf{d} are shown in
	Figure~\ref{f.Pmap}\textbf{e} and \textbf{f} respectively. The two limit cycles approach
	heteroclinic loops $\overline{\Gamma^0\bigcup\Gamma^-}$ and
	$\overline{\Gamma^0\bigcup\Gamma^+}$ 
	in Hausdorff distance as $K\to K_{HC}-0$ and $K\to K_{HC}+0$ respectively (see Figure~\ref{f.boa}).
	Both limit cycles are stable. Furthermore, numerically computed Poincar\'e maps do not show
	any effect of the heteroclinic bifurcation. These counterintuitive observations will
	be explained in the next section by the analysis of the nonlocal bifurcation taking place
	in \eqref{system}.

	\begin{figure}
		\textbf{a}\;		\includegraphics[width = .45\textwidth]{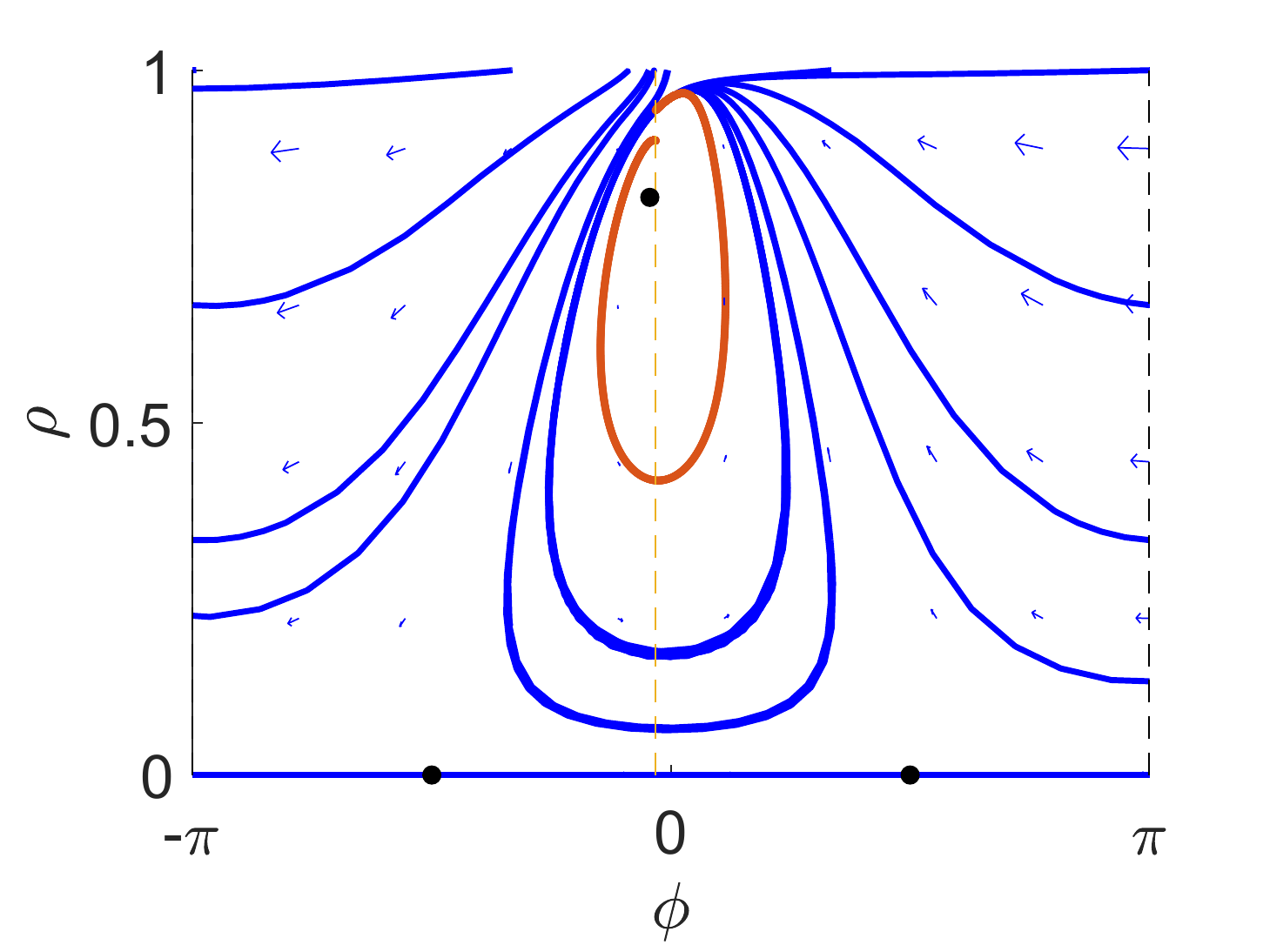}
		\textbf{b}\;	\includegraphics[width=.45\textwidth]{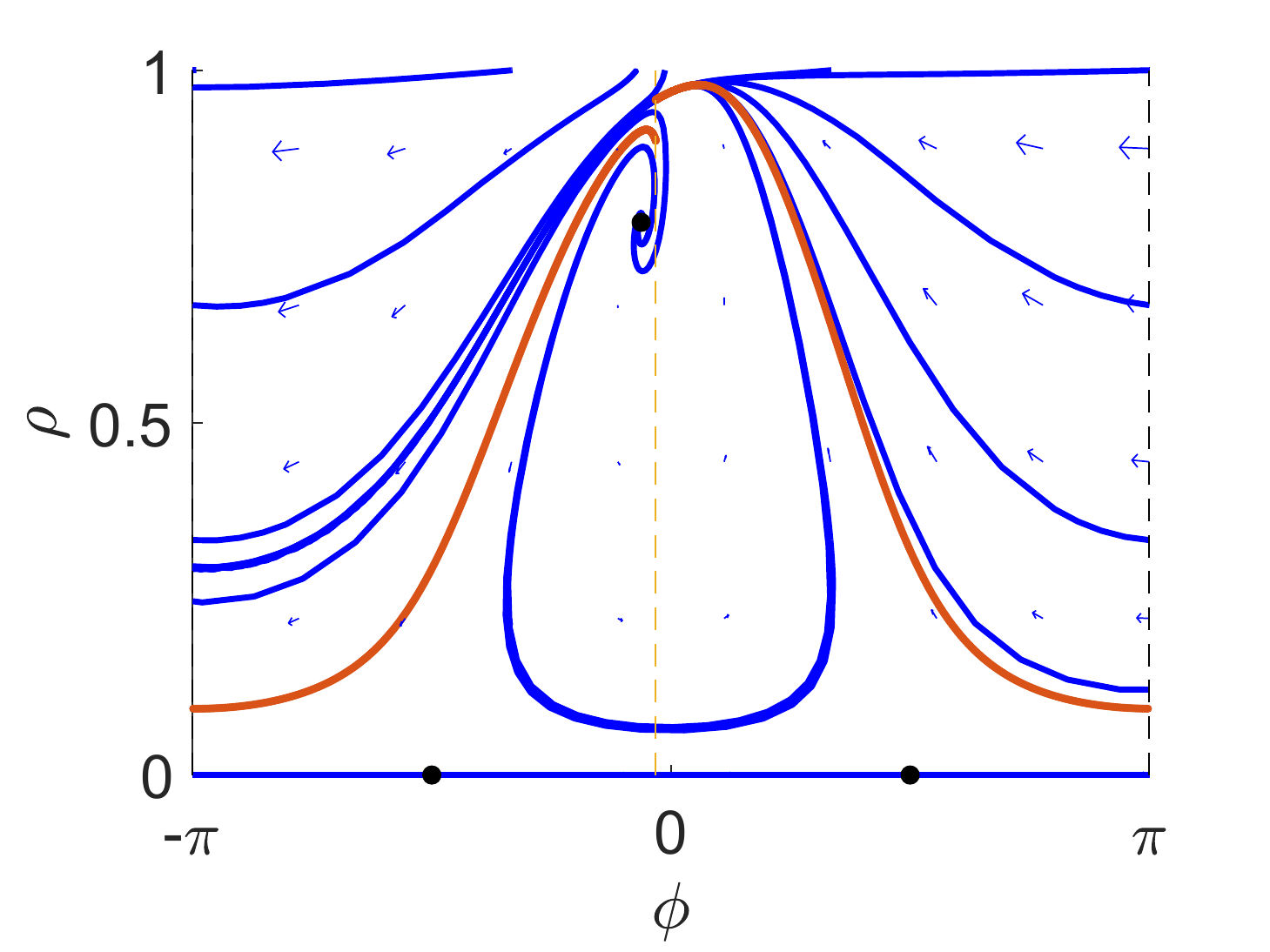}\\
		\textbf{c}\;\includegraphics[width = .45\textwidth]{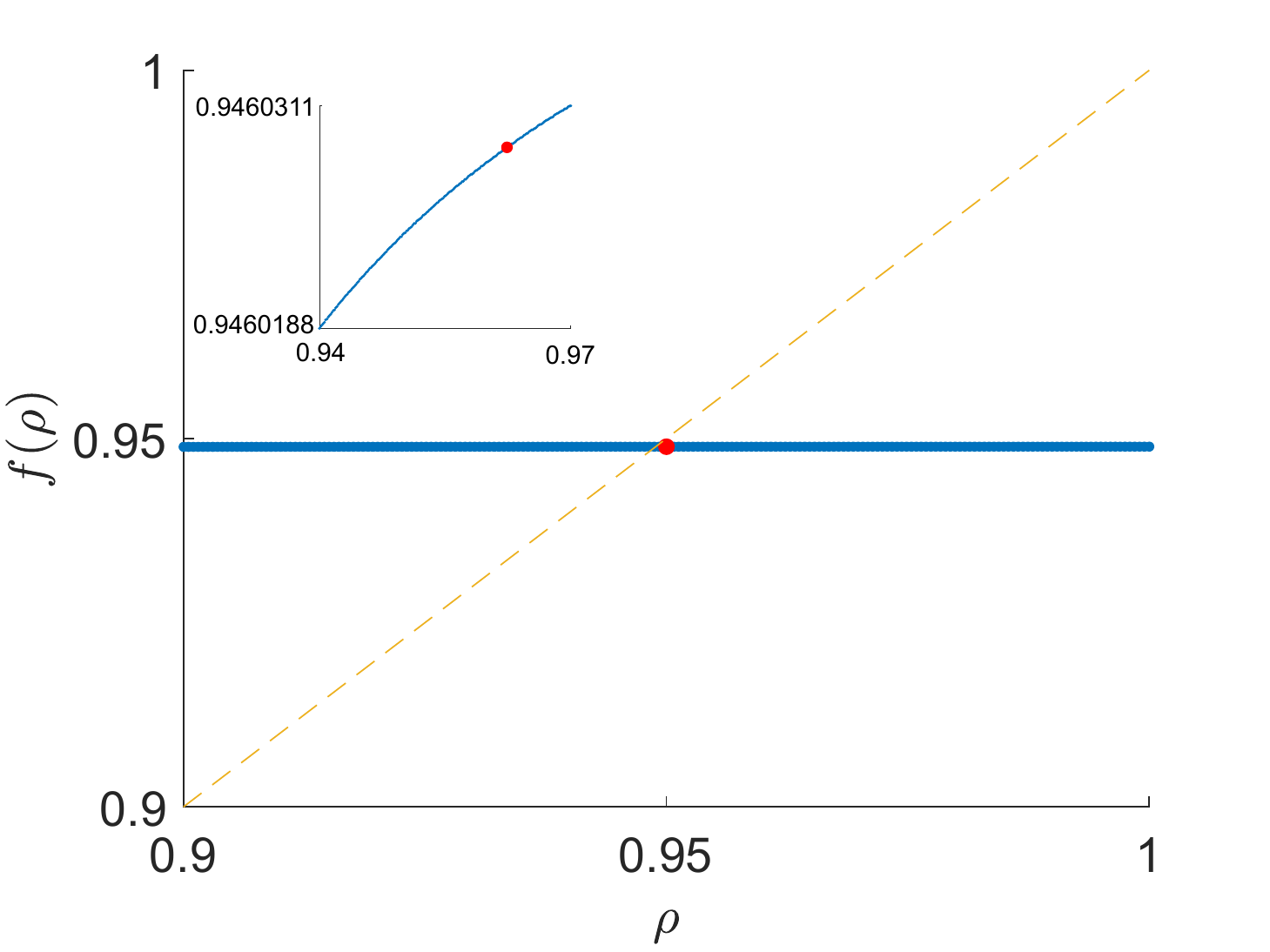}
		\textbf{d}\;	\includegraphics[width = .45\textwidth]{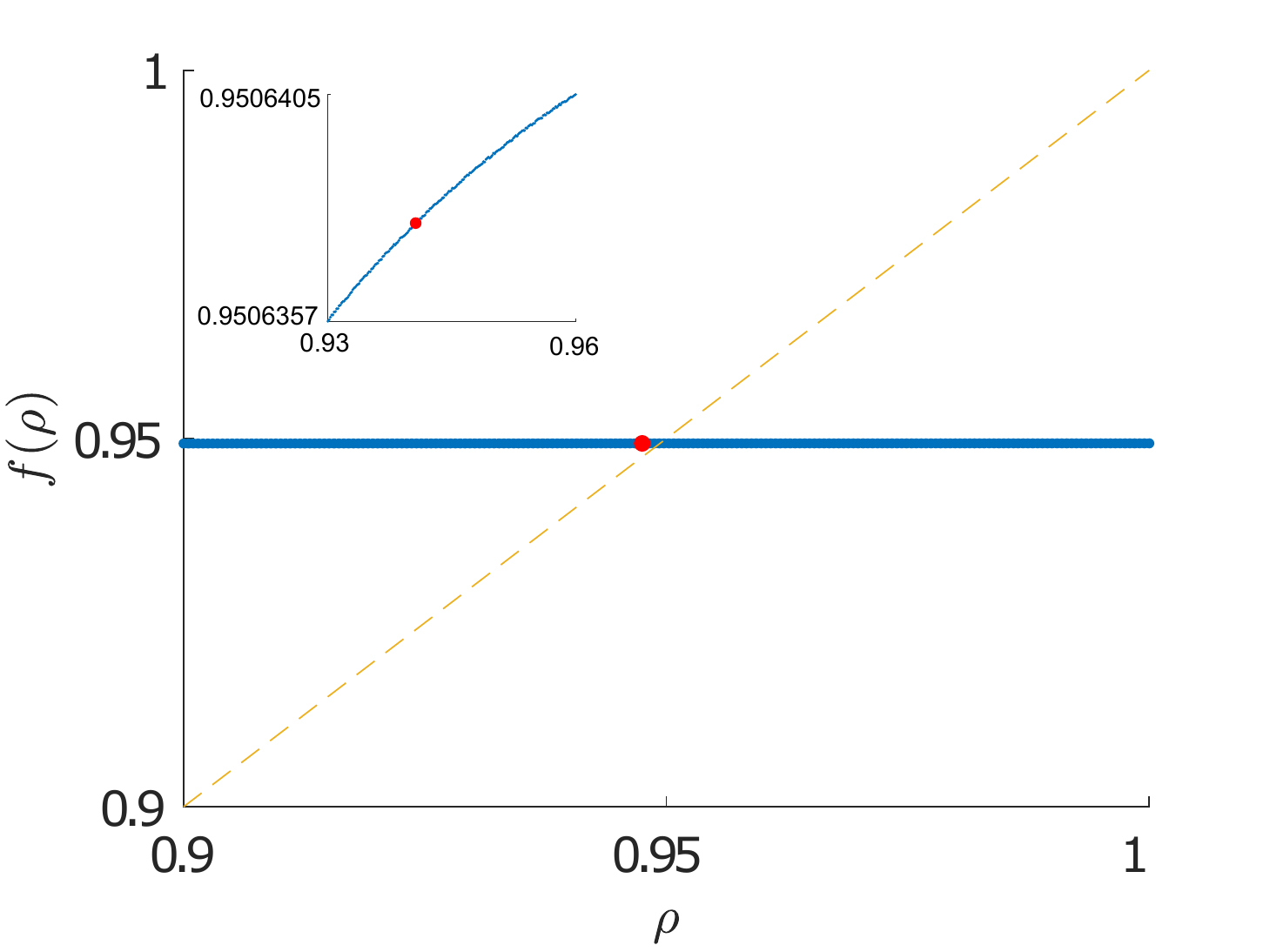}\\
		\textbf{e}\;	\includegraphics[width = .45\textwidth]{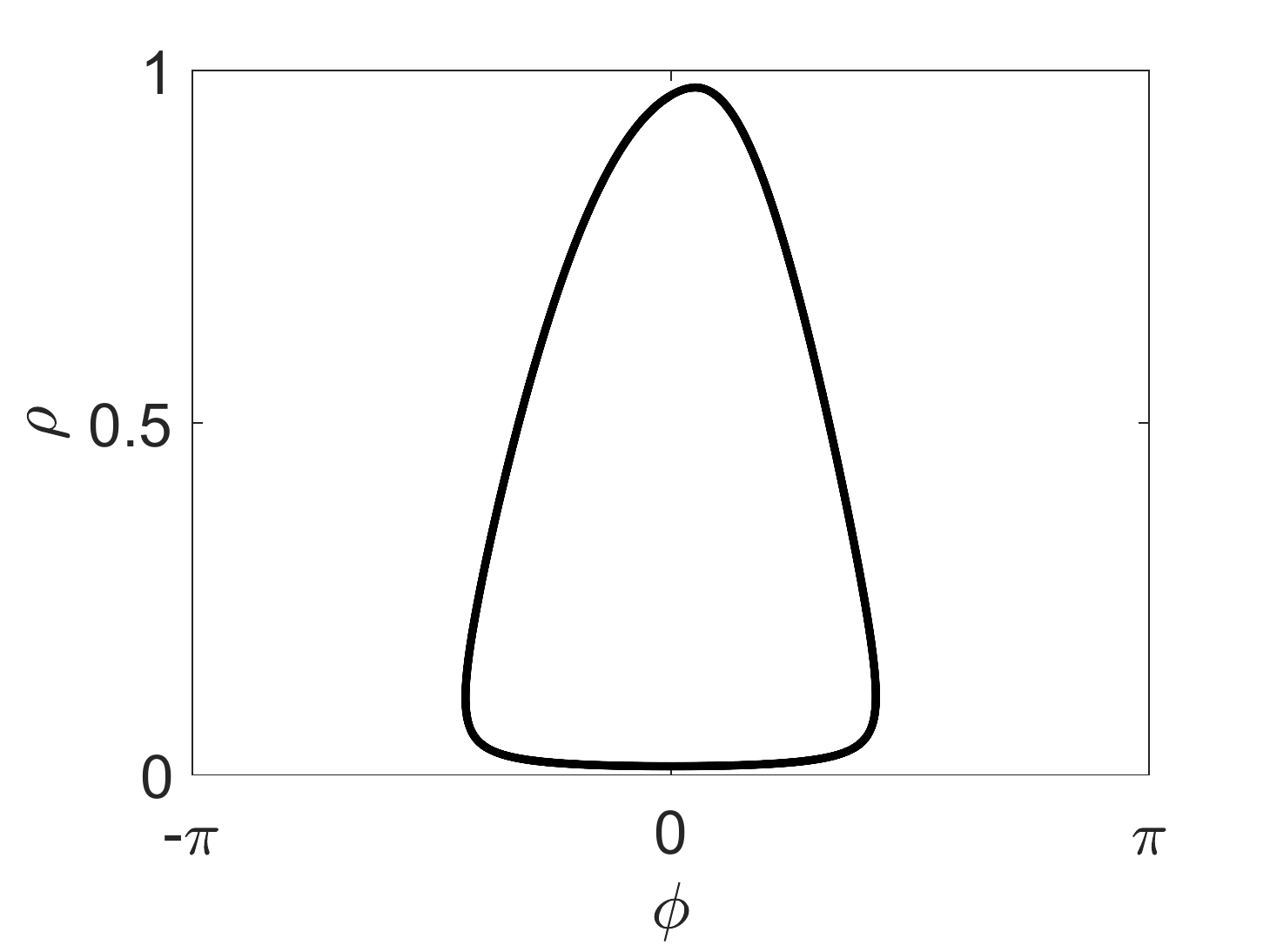}
		\textbf{f}\;	\includegraphics[width = .45\textwidth]{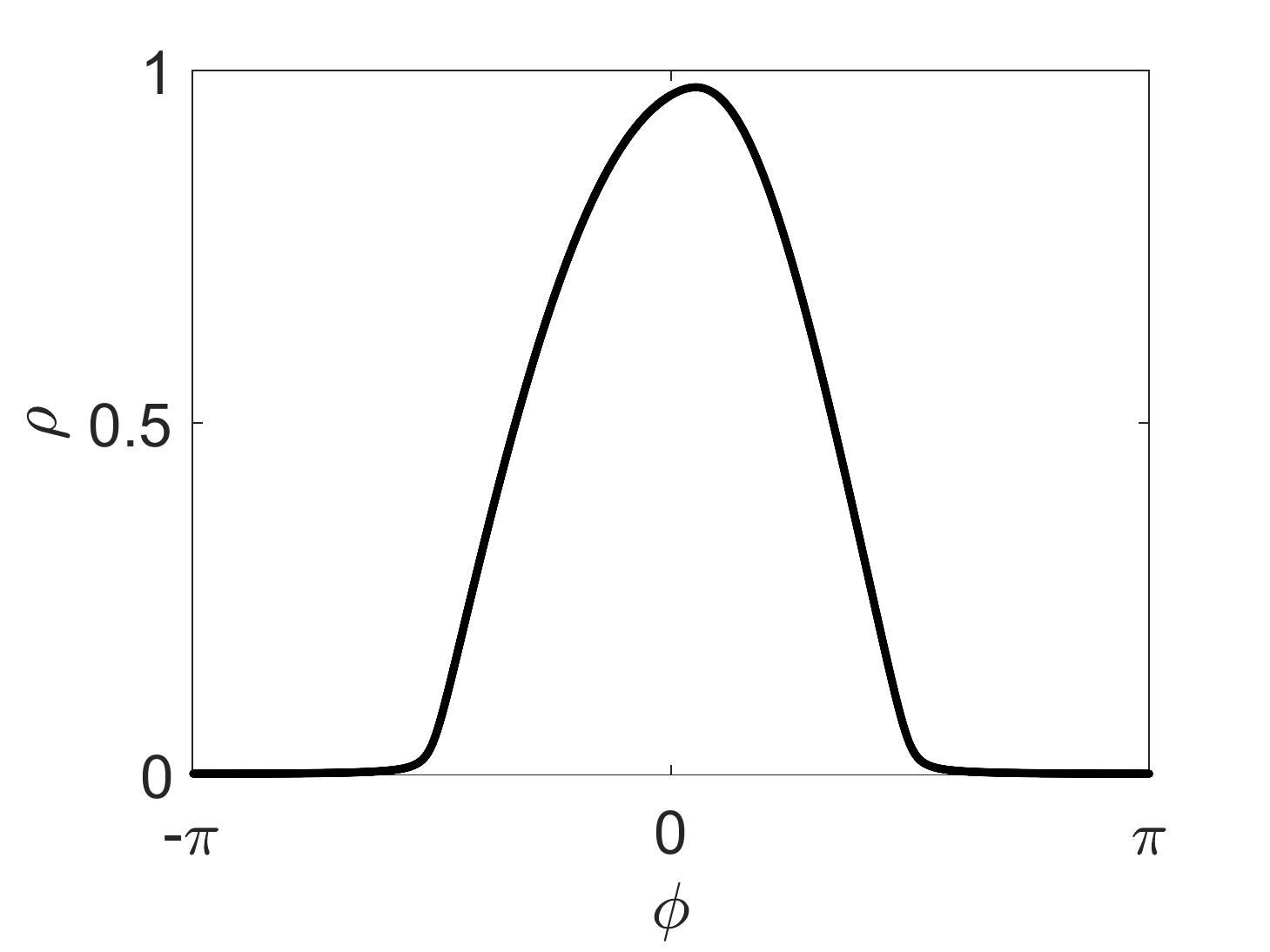}
		\caption{\textbf{a},\textbf{b}) The flow before and after the heteroclinic bifurcation ($K = 0.23$ and $K=0.3$).
			The Poincar\'e section, $\phi = -0.1$, is indicated by the
			dashed line. The red trajectories illustrate the construction of the Poincar\'e map from $\rho = 0.9$. 
			\textbf{c},\textbf{d}) The Poincar\'e map before and after the bifurcation: $K = 0.251$ and $0.253$. The red dot on the Poincar\'e section marks the point of intersection
			with $W^s(S_1)$. The Poincar\'e map is not defined at this point. 
			\textbf{e},\textbf{f}) The periodic solutions corresponding to the fixed points of the Poincar\'e  map before $(K=0.251)$ and after $(K=0.253)$ the bifurcations. Both orbits are stable. In all figures, $\epsilon = 0.1$ and $\delta = 0.01$.
		}
		\label{f.Pmap}
	\end{figure}

	\section{The boa constrictor bifurcation}
	\label{sec.heteroclinic}
	\setcounter{equation}{0}
	\begin{figure}
		\centering
		\includegraphics[width=.48\textwidth]{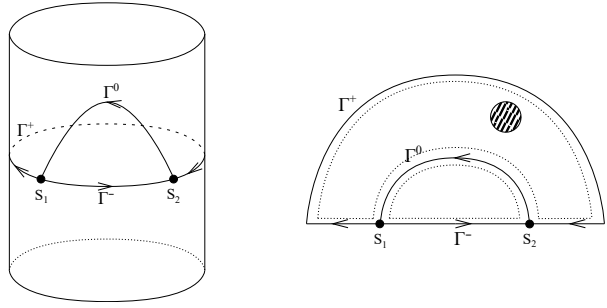}
		\caption{At at the heteroclinic bifurcation, two saddles, $S_1$ and $S_2$,
			are connected by three heteroclinic orbits $\Gamma^0$, $\Gamma^-$, and $\Gamma^+$.
			The sketch on the right explains the topology of the heteroclinic loops.}
		\label{f.boa}
	\end{figure} 
	
	The main result of this section is given in the following theorem.
	
	\begin{thm}\label{thm.boa}
		Consider a two-dimensional system on a cylinder
		\begin{equation}\label{2d}
			\dot x= f(x,\alpha), \; x\in (\R/2\pi \Z)\times\R, \alpha\in\R,
		\end{equation}
		with smooth $f$. Suppose that for all $\alpha\in\R$ there are two saddles $S_1$ and $S_2$
		connected with heteroclinic orbits $\Gamma^+$ and $\Gamma^-$
		such that $\overline{\Gamma^-\bigcup \Gamma^+}$ is a noncontractible simple closed curve.
		In addition, for $\alpha=0$ there is another heteroclinic orbit $\Gamma_0$ connecting $S_1$ and $S_2$
		such that $\overline{\Gamma^-\bigcup \Gamma_0}$ is a contractible simple closed curve (see Figure~\ref{f.boa}).
		
		Let $\lambda_{1,2}^s:=\lambda_{1,2}^s(0)<0<\lambda_{1,2}^u:=\lambda_{1,2}^u(0)$ be the eigenvalues of
		the Jacobian
		$Df(S_{1,2},0)$. Recall that $\sigma_{1,2}=|\lambda_{1,2}^s|/\lambda^u_{1,2}$ is called the
		saddle number. Further, we assume
		\begin{enumerate}
			\item
			$\sigma:=\sigma_1\sigma_2\neq 1$,
			\item
			$\beta^\prime(0)\neq 0$, where $\beta(\alpha)$ is a suitably defined
			split function (see below).
			
			Then for sufficiently small $|\alpha|$, there exists a small neighborhood of
			$$
			\Gamma:=\overline{\Gamma^-\bigcup \Gamma_0\bigcup\Gamma^+},
			$$
			which contains a unique limit cycle $\mathcal{P}_\beta$ bifurcating from $\Gamma$.
			Moreover, $\mathcal{P}_\beta$ is contractible for $\beta>0$ and noncontractible for $\beta<0$.
			It is stable if $\sigma>1$ and unstable otherwise.
		\end{enumerate}
	\end{thm}
	
	\begin{proof} The proof employs a standard scheme for analyzing global bifurcations,
		which  goes back to the proof of Andronov-Leontovich Theorem
		(cf.~\cite[Theorem 6.1]{Kuz-Bifurcations}).
		
		First, we introduce local cross sections $\Sigma_{1,2}$ and $\Pi_{1,2}^{\pm}$ near the
		saddles (see Fig.~\ref{f.return}). Next we construct the following flow-defined maps 
		\begin{align*}
			P_1: & \Sigma_1\to \Pi_1^-\cup\Pi_1^+,\\
			P_2: & \Sigma_2\to \Pi_2^-\cup\Pi_2^+,\\
			Q^-: & \Pi_1^-\to \Pi_2^-,\\
			Q^+: & \Pi_1^+\to \Pi_2^+,\\
			R: & \Sigma_2\to\Sigma_1.
		\end{align*}
		The near-to-saddle maps $P_1$ and $P_2$ capture the local dynamics near $S_1$ and $S_2$.
		$P_1$ and $P_2$ do not depend on $\alpha$  to leading order for small $|\alpha|$.
		Maps $Q^-$ and $Q^+$ are defined by the flow near $\Gamma^-$ and $\Gamma^+$.
		These maps do not depend on $\alpha$. Finally, $R$ is defined by the flow in a small vicinity
		of $W^u(S_2)$.
		
		On $\Sigma_1$ we select a system of coordinates such that $\xi=0$
		corresponds to the point of intersection of $\Sigma_1$ and
		$W^s(S_1)$. Then the coordinate $\xi$ of the point of intersection of
		$W^u(S_2)$ with $\Sigma_1$ defines the split function:
		$$
		\beta(\alpha):=\xi.
		$$

		The near-saddle map $P_1$ is computed from the following 
		linear system after a suitable change of coordinates near $S_1$ (cf.~\cite[Chapter 9]{IlyWei})
		\begin{equation}\label{saddle}
			\begin{split}
				\dot \xi & =\lambda_1^u \xi,\\
				\dot \eta&= \lambda_2^s\eta.
			\end{split}
		\end{equation}
		in the neighborhood of the origin with cross-sections $\Sigma_1=\{(\xi,1): \; \xi\in [-1,1]\}$,
		$\Pi_1^-= \{(1, \eta): \; \eta\in [0,1]\}$, and $\Pi_1^+= \{(-1, \eta): \; \eta\in [0,1]\}$ (see Fig.~\ref{f.saddle}).
		A standard computation yields
		$$
		P_1(\xi)=\xi^{\sigma_1}.
		$$
		Similarly, we compute
		$$
		P_2(\xi)=\xi^{\sigma_2}.
		$$
		The global maps are given by
		\begin{align}\label{compute-maps}
			Q^-(\eta) &= a^-\eta + O(\eta^2),\\
			Q^+(\eta) & =a^+\eta +O(\eta^2),\\
			R(\xi)& =\beta + a\xi +O(\xi^2),
		\end{align}
		where coefficients $a^-, a^+,$ and $a$ are positive and $\beta$ is the split function.

		We can now compute the first return map
		\begin{equation}\label{def-P}   
			P=
			\left\{
			\begin{array}{ll}
				R\circ P_2 \circ Q^- \circ P_1, & \beta>0,\\
				R\circ P_2 \circ Q^+ \circ P_1, & \beta<0.
			\end{array}
			\right.
		\end{equation}

		The combination of \eqref{compute-maps} and \eqref{def-P}  yields
		\begin{equation}\label{compute-P}   
			P(\xi)=
			\left\{
			\begin{array}{ll}
				\beta+ a(a^-)^{\sigma_2} \xi^{\sigma} +\mbox{higher order terms}, & \beta>0,\\
				\beta+ a(a^+)^{\sigma_2} \xi^{\sigma} +\mbox{higher order terms}, & \beta<0.
			\end{array}
			\right.
		\end{equation}
		From \eqref{compute-P}, one can see that $P$ has a fixed point $\bar\xi_\beta$ in a neighborhood
		of the origin for small $|\beta|\neq 0$. Further, $\bar\xi_\beta \beta >0$ for $\beta\neq 0$. 
		Finally, $\bar\xi_\beta$ is stable if $\sigma>1$ (dissipative case) and is unstable if $\sigma<1$.
	\end{proof}
	\begin{rem} If $\sigma=1$ as in the case of \eqref{system} the stability of the bifurcating
		orbits is determined by the positive coefficients $a$, $a^-$, and $a^+$. Specifically, $\mathcal{P}_\beta$
		is stable if $a(a^-)^{\sigma_2}<1$ for $\beta>0$ and  $a(a^+)^{\sigma_2}$ for $\beta<0$.
		In this case, the stability of the orbit is determined by the degree of contraction produced by the
		global maps defined by the flow
		along the heteroclinic loops rather than by that of the local near-to-saddle maps.
		While obtaining analytical estimates of the contraction of the global maps requires tedious calculations,
		numerical first return maps in Figure~\ref{f.Pmap}\textbf{c},\textbf{d} unequivocally demonstrate
		strong contraction.
	\end{rem}
	
	\begin{figure}[h]
		\centering
		\includegraphics[width=.48\textwidth]{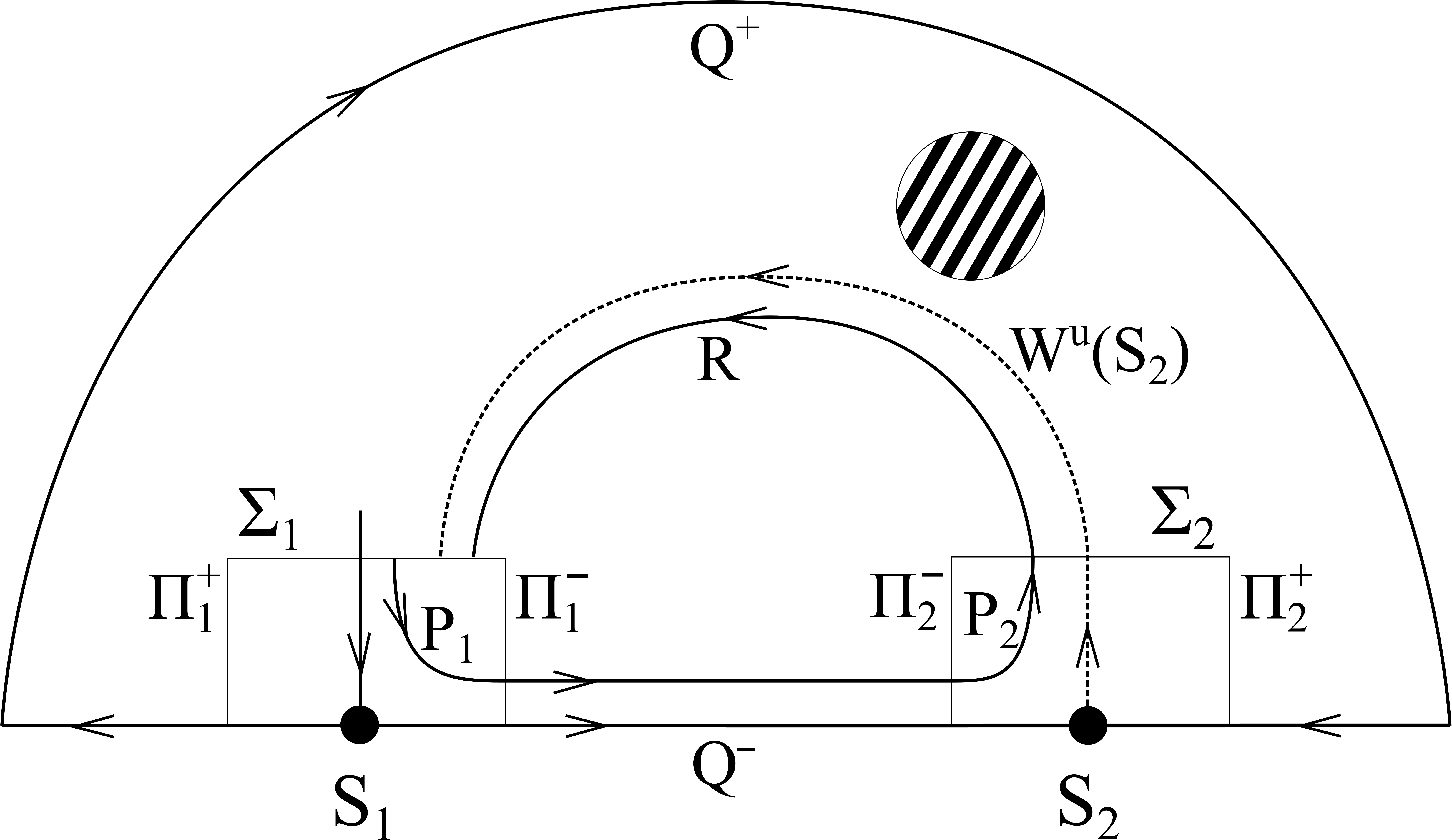}
		\caption{The first return map (see text for details).}
		\label{f.return}
	\end{figure} 
	
	\begin{figure}[h]
		\centering
		\includegraphics[width=.6\textwidth]{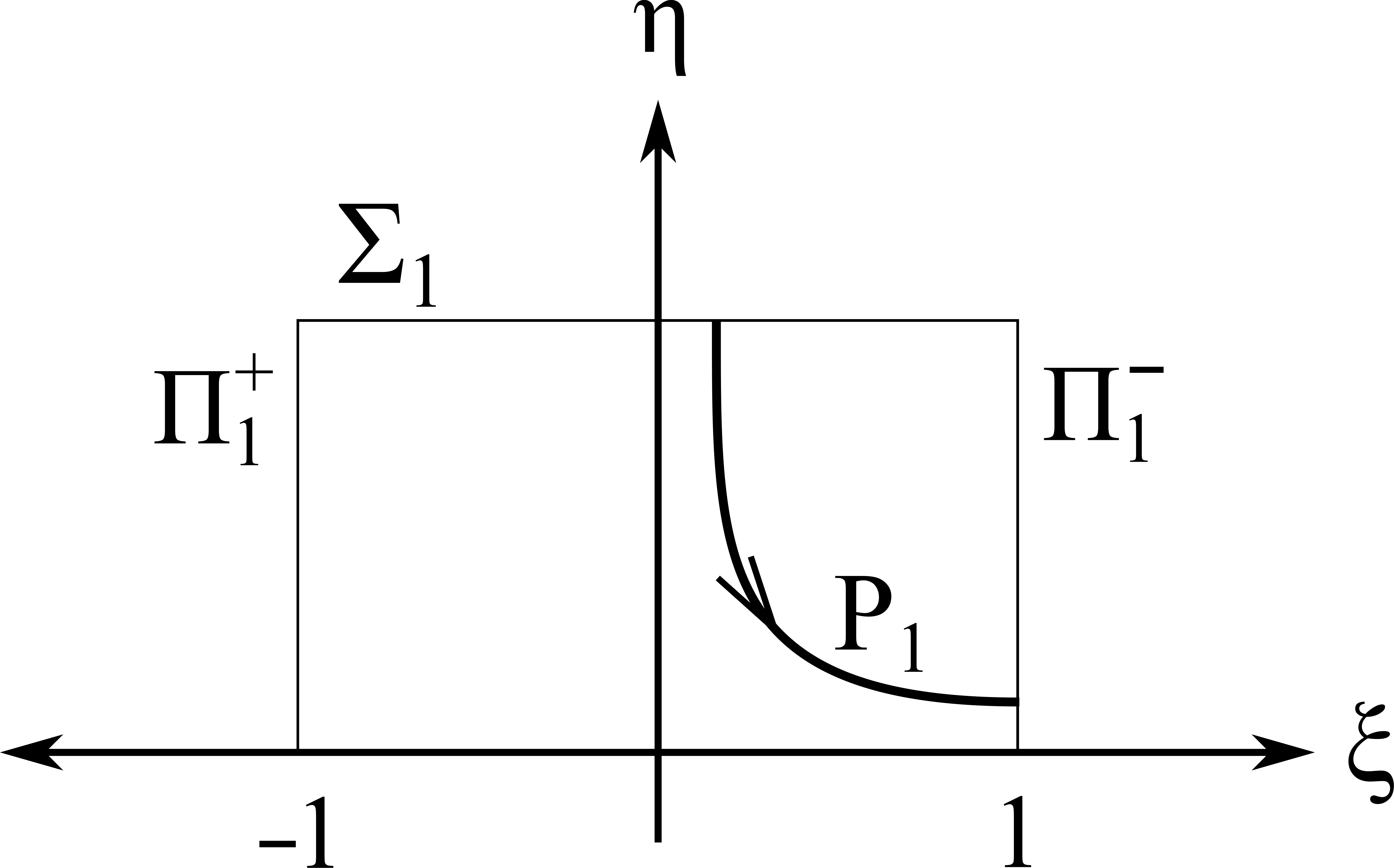}
		\caption{The local (near the saddle) component of the first return map.}
		\label{f.saddle}
	\end{figure}

	\section{Collective dynamics}\lbl{sec.collective}
	\begin{figure}
		\centering
		{\bf a}\includegraphics[width=	.45\textwidth]{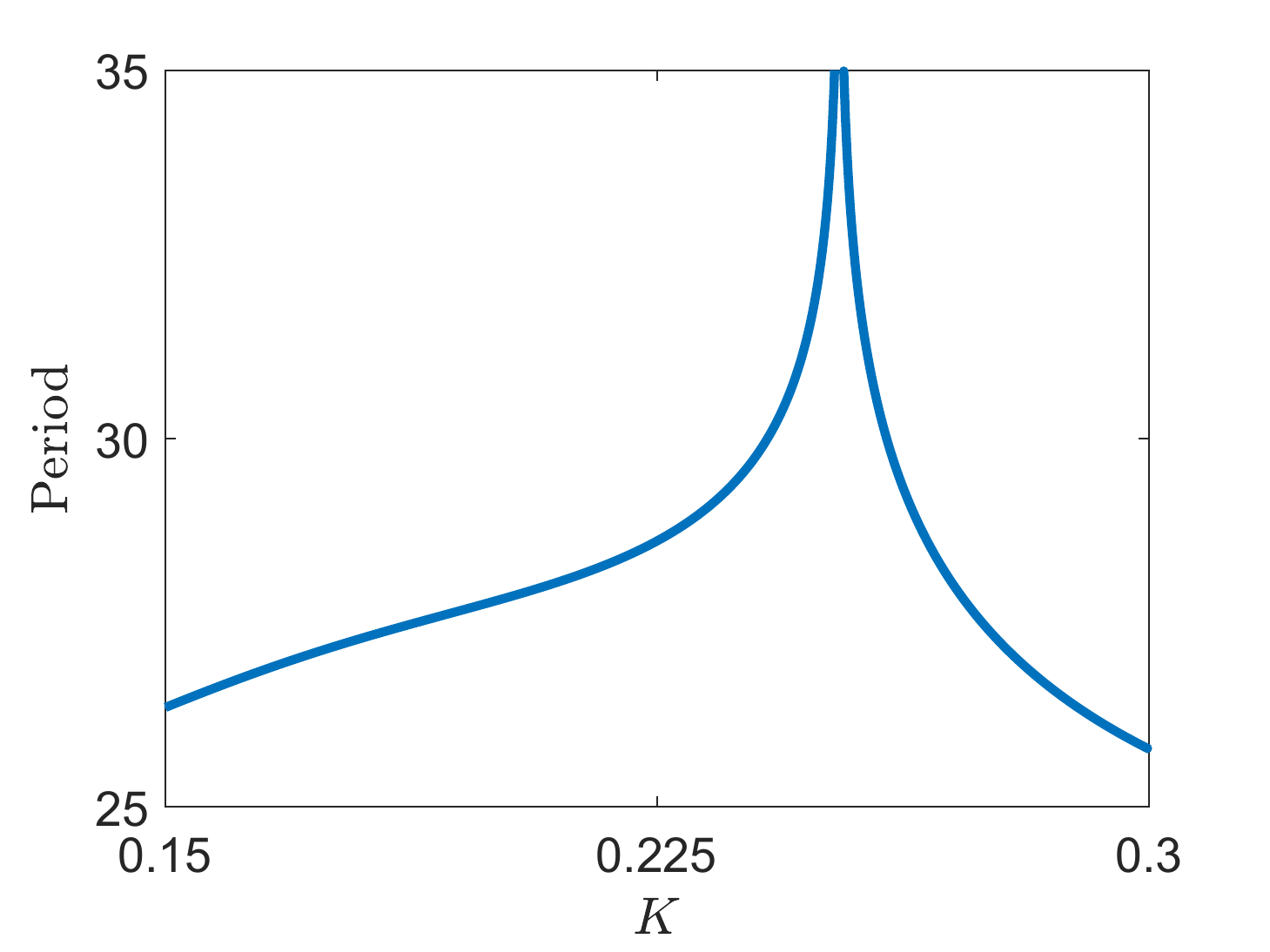}
		{\bf b}\includegraphics[width=	.45\textwidth]{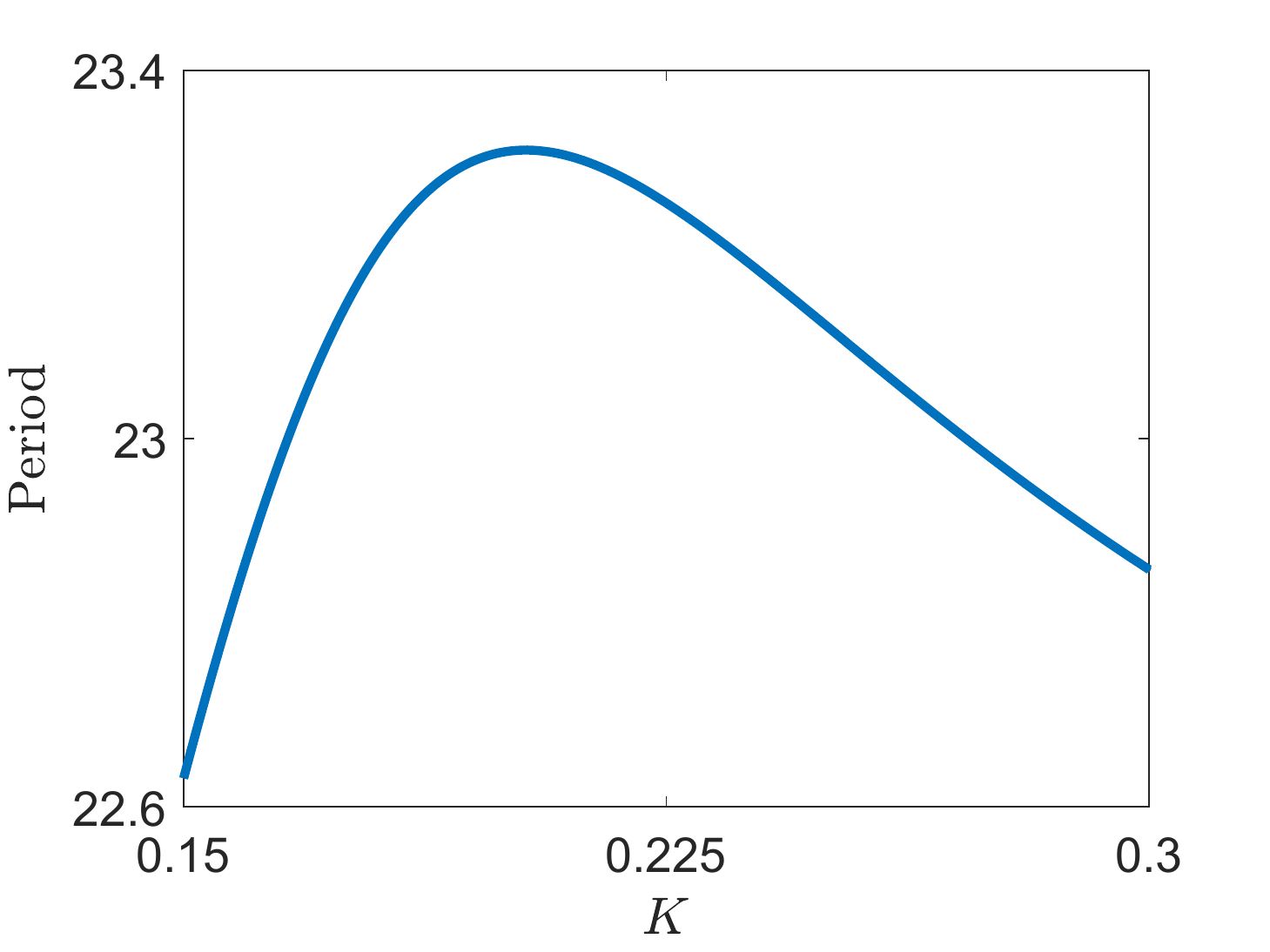}
		\caption{The period near the heteroclinic bifurcation for the systems with rescaled
			and original time (plots \textbf{a} and \textbf{b} respectively; cf. \eqref{system}
			and \eqref{2odes}). Parameters are $\epsilon = 0.1$ and $\delta = 0.01$.
		} 
		\label{f.period}
	\end{figure}
	\begin{figure}
		\centering    
		\textbf{a}\;	\includegraphics[width = .45\textwidth]{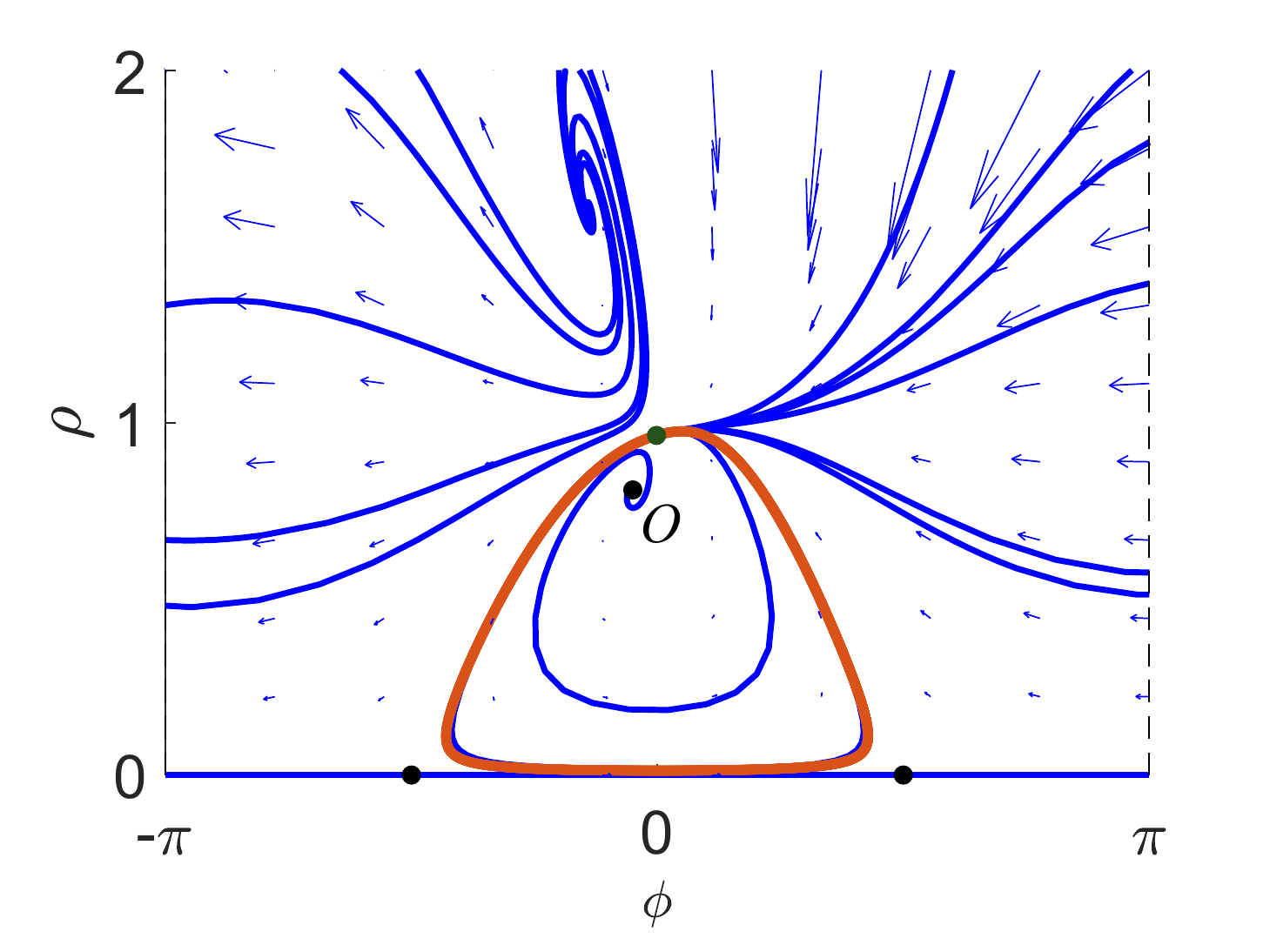}
		\textbf{b}\;	\includegraphics[width = .45\textwidth]{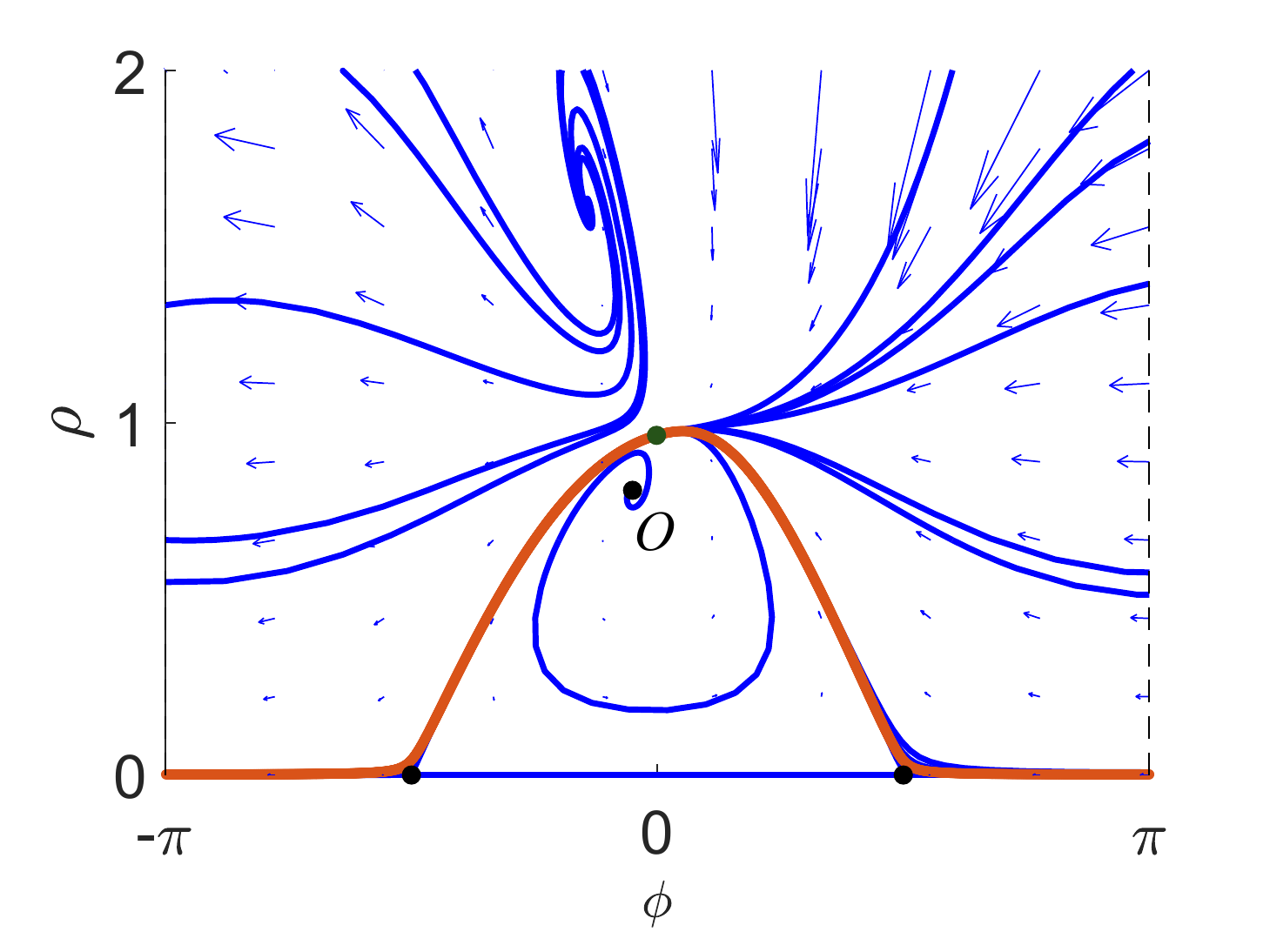}
		\caption{ The phase portraits before (\textbf{a}) ($K=0.251$) and after (\textbf{b}) ($K=0.253$) the heteroclinic
			bifurcation show
			that the periodic trajectory (shown in red) passes near the unstable focus
			$O$ in each case. $O$ is close to $\{\rho =1\}$ because the system is still
			close to the Bogdanov-Takens bifurcation. The green dot indicates
			the region where the periodic trajectory spends most of the period. Other parameters are chosen as above: $\epsilon = 0.1$, $\delta = 0.01$.
		} 
		\label{f.portraits}
	\end{figure}
	
	Having analyzed the reduced system \eqref{system}, we now return to the description
	of the collective dynamics of \eqref{KM}. The analysis in the previous two sections shows
	that the heteroclinic bifurcation separates two topologically distinct families of the limit
	cycles of \eqref{system} ${\mathcal P}_K$ for $K<K_{HC}$ and $K>K_{HC}$. In either case,
	${\mathcal P}_K$ consists of a segment lying in $V_c=\{0\le \rho<c\}$ for some $0<c\ll 1$
	and an arc outside $V_c$ (see Fig.~\ref{f.boa}). Thus, in one cycle of oscillations
	the composition of the population of oscillators changes from high incoherence ($\rho\approx 0$)
	to high coherence ($\rho\approx 1$).
	
	To estimate the duration of each of these phases and the period of oscillations note that
	the incoherent phase is determined by the time the periodic trajectory spends near the saddles,
	which can be easily estimated from \eqref{saddle}. The trajectory of \eqref{saddle} starting
	from $(\beta, 1)$ leaves the strip $|\xi|\le 1$ after time
	\begin{equation}\label{T1}
		T_1\simeq \frac{1}{\lambda_1^u} \ln\frac{1}{|\beta|}.
	\end{equation}
	Furthermore, at the time of exit
	\begin{equation}\label{exit-xi}
		\eta=|\beta|^{\sigma_1}.
	\end{equation}
	Thus, the trajectory enters the neighborhood of $S_2$ with $\eta=O\left(|\beta|^{\sigma_1}\right).$
	From this, we can estimate the time it spends in the vicinity of $S_2$:
	\begin{equation}\label{T2}
		T_2\simeq \frac{\sigma_1}{\lambda_1^u} \ln\frac{1}{|\beta|}.
	\end{equation}
	From \eqref{T1} and \eqref{T2}, we  estimate the period of oscillations near the heteroclinic
	bifurcation:
	\begin{equation}\label{period}
		T\simeq \max\left\{\frac{1}{\lambda_1^u} \ln\frac{1}{|\beta|},
		\frac{\sigma_1}{\lambda_2^u} \ln\frac{1}{|\beta|}\right\}.
	\end{equation}
	(see Fig.~\ref{f.period}\textbf{a}).
	%
	
	After going back to the original time, for the oscillations in \eqref{2odes} we obtain:
	\begin{flalign*}
		T_1^\prime &= \int_0^{T_1} \rho(\tau) d\tau \simeq \int_0^{T_1} e^{\lambda_1^s \tau}d\tau\\
		&\simeq \frac{1-C_1 |\beta|^{\sigma_1}}{|\lambda_1^s|},
	\end{flalign*}
	and similarly
	\begin{equation*}
		T_2^\prime \simeq \frac{1-C_1 |\beta|^{\sigma}}{|\lambda_2^s|}.
	\end{equation*}
	Thus, in the original time the period of oscillations remains
	finite (see Fig.~\ref{f.period}~\textbf{b}). Furthermore, since the points on $W^s(S_1)$ hit $S_1$
	in finite time,  the heteroclinic orbit reaches the saddle in finite time too.
	Recall that $\rho=0$ is the image of the origin in $\R^2$
	under the polar coordinate transformation. The origin is not a fixed point of \eqref{cart}. Thus,
	the heteroclinic orbit of \eqref{system} for $K=K_{HC}$ corresponds to a periodic orbit
	of \eqref{cart} passing through the origin in the cartesian coordinates (Figure~\ref{f.cart}).
	Thus, the family of periodic orbits $\mathcal{P}_K, K>K_{AH}$ of \eqref{system}
	corresponds to a family of periodic orbits $\mathcal{Z}_K$ of \eqref{cart}.
	The period of $\mathcal{Z}_K$ achieves its maximum at $K=K_{HC}$
	(see Fig.~\ref{f.period}~\textbf{b}). The discrepancy in different locations of the maxima
	in plots \textbf{a} and \textbf{b} in Figure~\ref{f.period} is explained by the fact that
	we dropped higher order terms in \eqref{BT-delta-epsilon}.

	Having described the `incoherent' portion of the periodic orbit lying in a neighborhood
	of $\{\rho=0\}$, we now turn to the complimentary portion lying in a neighborhood
	of $\Gamma_0$ (see Fig.~\ref{f.portraits}). First note that outside a small neighborhood of $\{\rho=0\}$
	the rescaling of time no longer has a qualitative impact on the system's dynamics.
	Further, on $\rho=1$ the vector field of \eqref{system} is pointed downward
	(see the $\rho$-equation in \eqref{system}). On the other hand, the unstable
	focus $O$ has not moved too far from $\{\rho=1\}$, where it emerged at the Bogdanov-Takens
	bifurcation (see Figure~\ref{f.portraits}). Thus, $\Gamma_0$ has to pass in a region
	between $\{\rho=1\}$ and
	the unstable focus $O$ (Figure~\ref{f.portraits}). This forces
	$\Gamma_0$ to pass close to $O$ where the vector field is very weak. Consequently,
	a significant portion of the period is spent near $\{\rho=1\}$, which correspond
	to the `coherent' phase of oscillations. This results in an interesting scenario,
	for which the separation of the
	timescales in oscillations of the order parameter just before and after the
	heteroclinic bifurcation (Figure~\ref{f.order}) is not due to the proximity of the
	heteroclinic bifurcation, as one would be tempted to assume, but rather to the
	proximity to the Bogdanov-Takens bifurcation. Thus, both the Bogdanov-Takens
	and the heteroclinic bifurcations have an impact on the qualitative features of the
	collective dynamics. 
	\begin{figure}
		\centering    
		\textbf{a}\,\includegraphics[width = .3\textwidth]{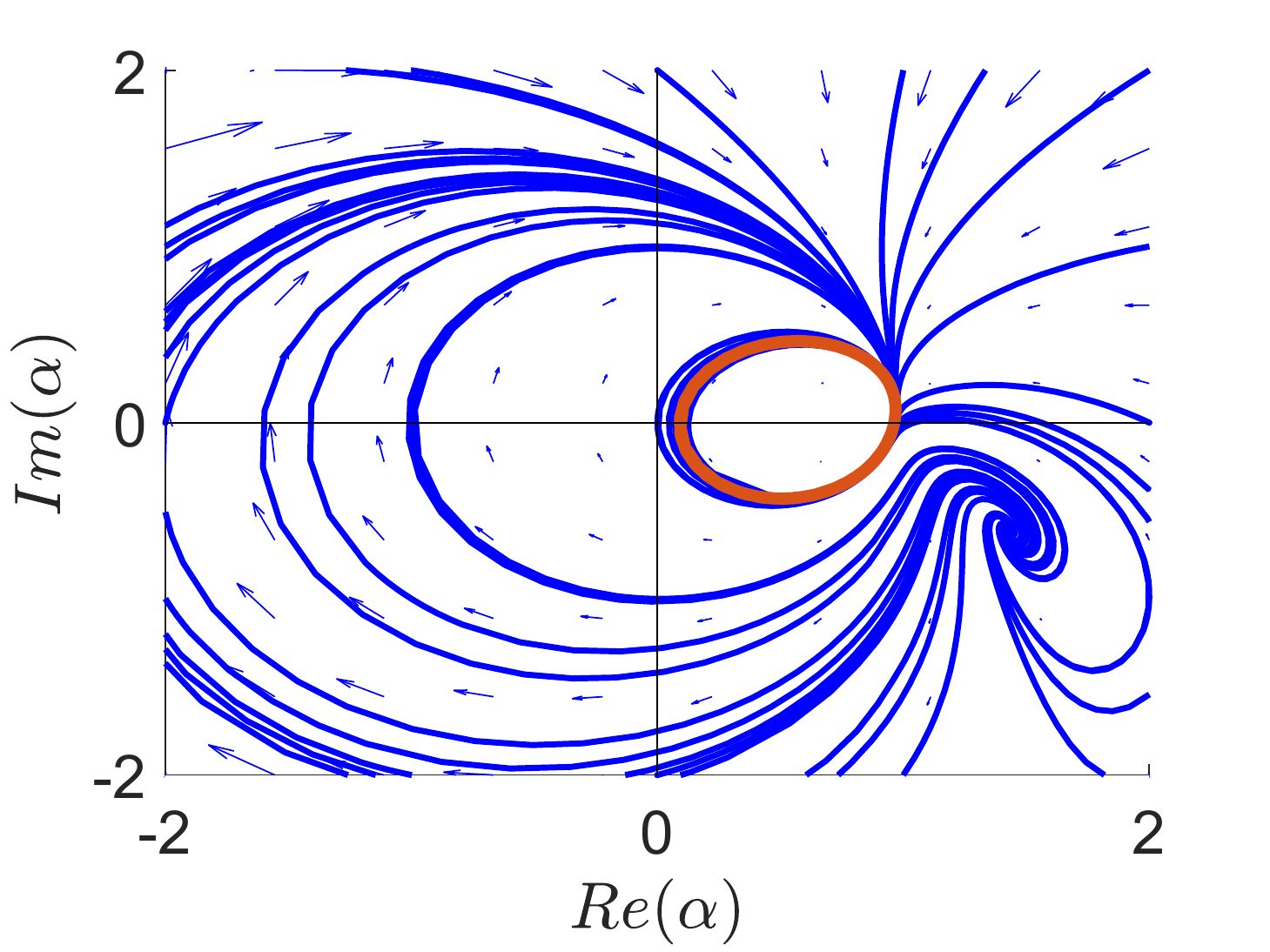}
		\textbf{b}\,\includegraphics[width = .3\textwidth]{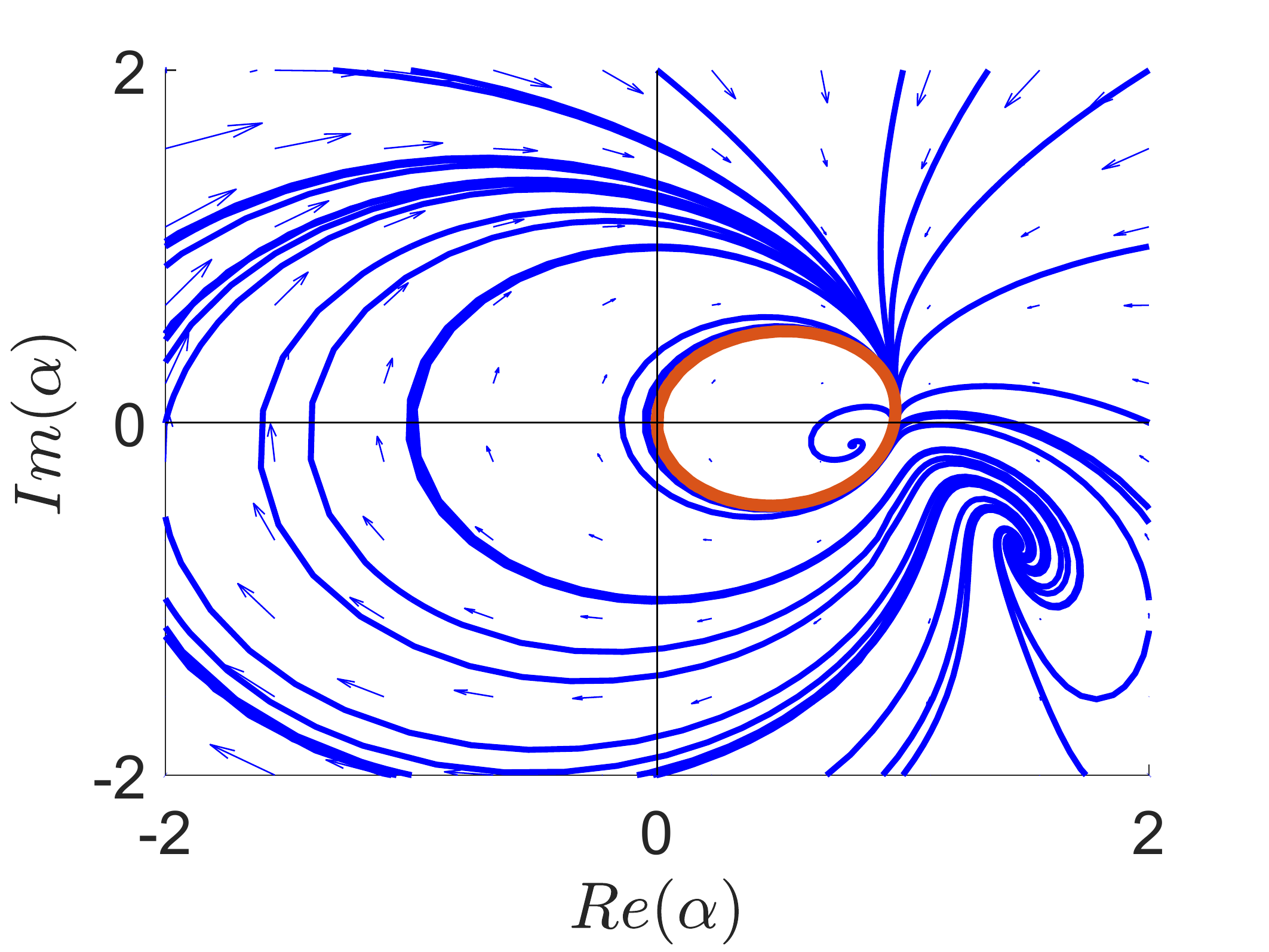}
		\textbf{c}\,\includegraphics[width = .3\textwidth]{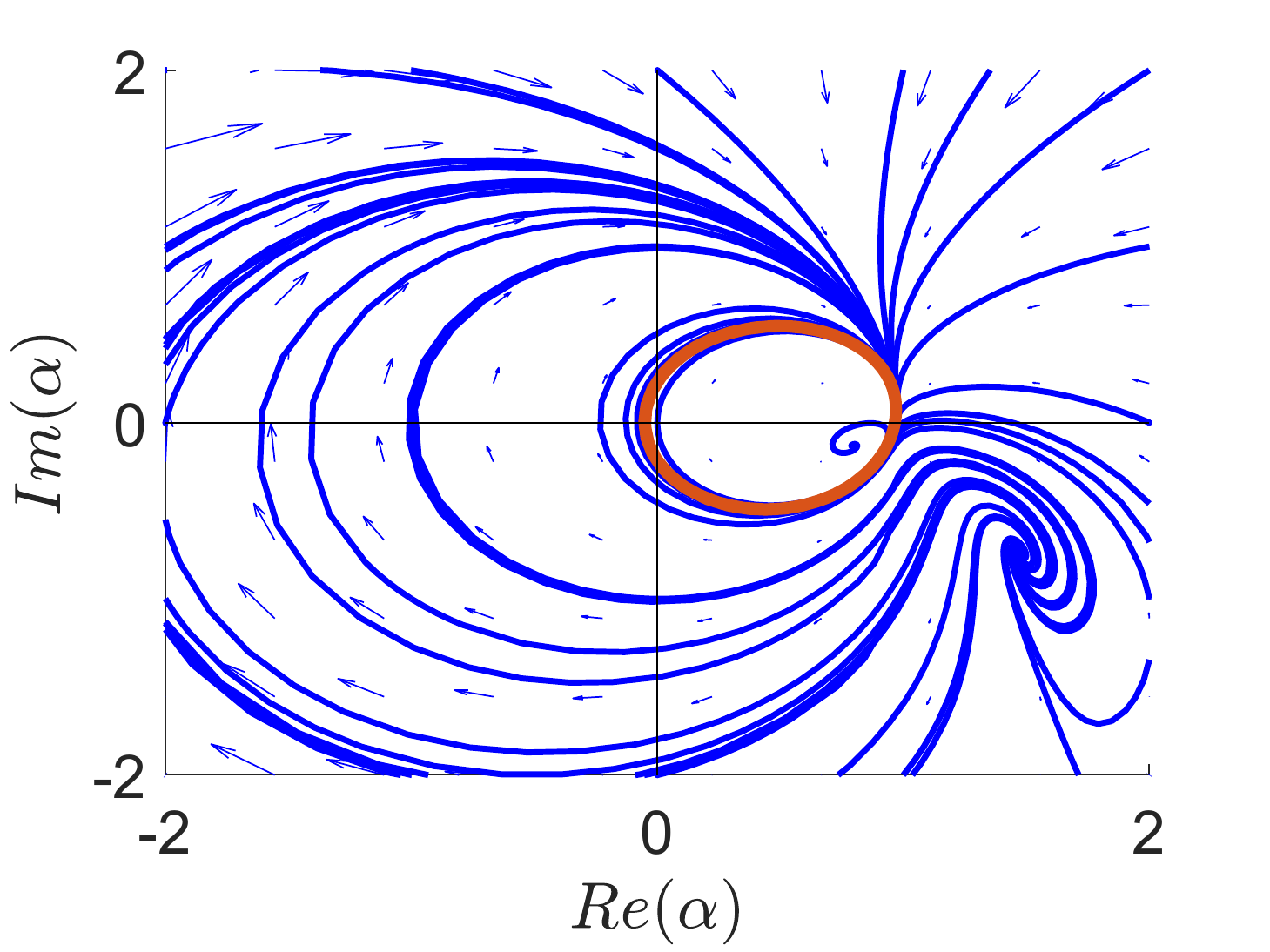}
		\caption{ The phase portraits of \eqref{cart} plotted in
			Cartesian coordinates. The middle plot (\textbf{b}) shows the periodic
			orbit passing through the origin. The two outer plots show
			the periodic orbits before (\textbf{a}) and after (\textbf{c})
			this event. Parameters are chosen as above: $\epsilon = 0.1$ and $\delta = 0.01$, with $K = 0.24, 0.253, 0.26$ from left to right.
		} 
		\label{f.cart}
	\end{figure}
	
	We conclude with several remarks on the relation between
	\eqref{cart}  and \eqref{system}. Recall that we switched to polar coordinates to be
	able to track the modulus and the argument of the order parameter, which capture
	macroscopic dynamics. The polar coordinate transformation is a diffeomorphism
	of $\R^2/\{(0,0)\}$ to $\R^+\times \T$, but its is not a bijection at the origin.
	Consequently \eqref{cart} and
	\eqref{system} are topologically equivalent only when restricted to $\R^2/\{(0,0)\}$ and
	$\R^+\times \T$ respectively. In particular, the Andronov-Hopf and Bogdanov-Takens bifurcations
	describing local transformations of the vector field \eqref{system} in  closed domains
	of $\R^+\times \T$ translate automatically to the corresponding bifurcations of 
	\eqref{cart}. On the other hand the heteroclinic bifurcation, which we analyzed for
	\eqref{system} involves the set $\rho=0$, which lies outside $\R^+\times \T$.
	The heteroclinic
	orbit connecting $\Gamma^0$ (Figure~\ref{f.boa}) corresponds to a periodic orbit
	of \eqref{cart} passing through the origin (Figure~\ref{f.cart}).
	This not a bifurcation of \eqref{cart} as a vector field on $\R^2$, but it is a bifurcation
	on $\R^2/\{(0,0)\}$. This  is a bifurcation, because a contractibe periodic orbit
	before hitting the origin becomes a noncontractibe one after this event.
	This is a border collision bifurcation.
	Clearly, this bifurcation is a consequence of our using the polar coordinate transformation,
	which is singular at the origin. Nonetheless, both the heteroclinic and the border collision
	bifurcations are relevant in the context of the macroscopic dynamics. While in the cartesian
	coordinates, the passing of the periodic orbit
	of \eqref{cart} through the origin is a regular event, it represents a transition
	point in the description of the macroscopic dynamics. This is the point where the oscillations
	of the center of mass of the population of oscillators are transformed into rotations. At this
	point the amplitude and the period of the oscillations of the order parameter
	reach their respective maximal values, while the
	modulus of the order parameter reaches its minimal value.  Note that the fast dips
	in the modulus of the order parameter can get arbitrarily close to $0$ provided that
	$n$ is large enough  (Figure~\ref{f.boa}\textbf{a}). At the transition point, the oscillators
	are most dispersed as they undergo fast transitions between their successive stays near  $\rho=1$
	(Figure~\ref{f.boa}\textbf{a}). Therefore, the use of the polar coordinates in the analysis
	of \eqref{cart} and the analysis of the heteroclinic bifurcation in the transformed system
	are essential for understanding macroscopic dynamics of the coupled system \eqref{KM}.


	\section{Discussion}\label{sec.discuss}
	\setcounter{equation}{0}

	In the present paper, we analyzed a modified KM with individual
	oscillators in the regime near a saddle-node on an invariant circle bifurcation.
	The modified model features a new type of collective dynamics with alternating
	phases of high and low coherence. Furthermore, the order parameter exhibits
	slow-fast oscillations, which reveal a pronounced separation of
	timescales in
	collective dynamics. For the most part
	of the period the value of the order parameter is close to $1$ corresponding to
	the coherent phase. The long periods of coherence are punctuated by brief intervals of
	highly incoherent collective dynamics. The Ott-Antonsen reduction and a careful
	analysis of the reduced system show that the salient features of the collective
	dynamics are explained by the model's proximity to a local
	Bogdanov-Takens bifurcation
	and a nonlocal heteroclinic bifurcation.
	In contrast to previously studied cases,  where one limit cycle or two limit cycles of opposite
	stability appear in a bifurcation of a
	homoclinic/heteroclinic contour (cf.~\cite{AAIS, Kuz-Bifurcations,
		KuzHoo2021, Dukov2018}), the heteroclinic bifurcation for the system at hand generates a 
	limit cycles on each side of the bifurcation. These limit cycles are
	topologically distinct (contractible versus noncontractible) and are either
	both stable or both unstable.
	
	\vskip 0.2cm
	\noindent
	{\bf Acknowledgements.} This work grew out of AP's Research Co-op at Drexel
	University. GSM and AP were supported in part by NSF grant DMS 2009233 (to GSM).
	MSM was supported by a Support of Scholarly Activities Grant at The College of New Jersey.
	
	\noindent\textbf{Data availability statement.} Data sharing is not
	applicable to this article as no datasets were generated or analyzed
	during the current study.

	\bibliographystyle{amsplain}
	\def\cprime{$'$} \def\cprime{$'$}
	\providecommand{\bysame}{\leavevmode\hbox to3em{\hrulefill}\thinspace}
	\providecommand{\MR}{\relax\ifhmode\unskip\space\fi MR }
	\providecommand{\MRhref}[2]{%
		\href{http://www.ams.org/mathscinet-getitem?mr=#1}{#2}
	}
	\providecommand{\href}[2]{#2}

\end{document}